\documentclass[12pt]{article}

\usepackage[margin=1in]{geometry}
\usepackage{amsmath, amssymb, amsthm, mathtools}
\usepackage{bm}                   
\usepackage{graphicx}
\usepackage{booktabs}
\usepackage{hyperref}
\usepackage[round,authoryear]{natbib}
\usepackage{microtype}
\usepackage{xcolor}
\usepackage{setspace}
\usepackage{algorithm}
\usepackage{algpseudocode}
\usepackage{tikz}
\usetikzlibrary{arrows.meta, calc, decorations.pathreplacing,
                decorations.markings}
\usepackage{subcaption}

\newtheorem{definition}{Definition}
\newtheorem{proposition}{Proposition}
\newtheorem{remark}{Remark}

\newcommand{\R}{\mathbb{R}}
\newcommand{\E}{\mathbb{E}}

\newcommand{\diag}{\mathrm{diag}}

\newcommand{\norm}[1]{\left\|#1\right\|}

\newcommand{\Bern}{\mathrm{Bernoulli}}
\newcommand{\Normal}{\mathcal{N}}
\newcommand{\GFF}{G_{FF}}
\newcommand{\GBF}{G_{BF}}
\newcommand{\GBB}{G_{BB}}
\newcommand{\Bs}{\mathcal{B}}
\newcommand{\Fs}{\mathcal{F}}
\newcommand{\Tot}{\mathcal{E}}
\newcommand{\Aconn}{A}

\title{\textbf{A Flat Connection:\\
The Pooling Factor and the Geometry of Centring in Hierarchical MCMC}%
\thanks{Code and package available at \url{https://github.com/ABindoff/fibr}.}}

\author{Aidan D.\ Bindoff \\
  University of Tasmania \\
  \texttt{aidan.bindoff@utas.edu.au}}

\date{}

\begin{document}

\maketitle


\begin{abstract}
Standard MCMC diagnostics ($\hat{R}$, effective sample size, divergence
counts) detect whether a chain has mixed, but not why it has not.
We ask whether the centring/non-centring obstruction in hierarchical models has
a \emph{geometric} cause beyond the metric.
The joint parameter space is a fiber bundle (hyperparameters the base,
group-level parameters the fibers), and the Fisher information metric induces an
Ehresmann connection $A = -\GFF^{-1}\GBF$; the natural hypothesis is that the
obstruction is its curvature, felt by the sampler as holonomy.
We prove this false.
The connection is \emph{flat} for \emph{any} smooth hierarchical posterior,
not only the Gaussian or conjugate case, because its horizontal leaves are
the level sets of the fiber score $\partial_{\bm\alpha}\log p$: there is no
geometric obstruction above the metric.
What remains is statistical, not geometric, and the flat connection identifies it
as a single quantity: the conditional dependence of fiber on base, governed per
group by the \emph{prior fraction} $\pi_j$, the classical pooling factor. From it
the framework recovers the established picture, that the prior-dominated groups
are the ones that mix slowly and that the optimal per-group non-centring weight
follows in closed form, and a simulation study then separates this base--fiber
coupling from the funnel, a distinct base-space pathology, by their opposite
dependence on the hierarchical variance.
A direct attribution test confirms that NUTS does not transport the fiber: the
chain-level footprint is excess conditional autocorrelation in prior-dominated
groups, exactly as $\pi_j$ predicts.
Genuine, even rotational, curvature \emph{does} appear --- but only for
connections built from a sampler's \emph{working} metric (a fixed mass matrix),
where holonomy re-enters as an algorithmic rather than geometric phenomenon.
The prior-fraction diagnostic is distributed as the open-source R package
\texttt{fibr}, with the geometric methods as accompanying reproduction code.
\end{abstract}


\newpage
\tableofcontents
\newpage

\section{Introduction}
\label{sec:intro}

A central challenge in Bayesian computation is the hierarchical model.
The joint parameter space of a two-level model $-$ hyperparameters
$(\mu, \sigma)$ and group-level parameters $(\alpha_1, \ldots, \alpha_J)$
$-$ has a geometry that makes it simultaneously natural to write and
difficult to sample from.
The pathology is well documented: the centred parameterisation develops
a funnel, non-centred parameterisations help when the prior dominates,
and the right choice depends on a balance of prior and likelihood information
that varies by group and dataset.

A geometric account of this trade-off already exists, at the level of the
\emph{metric}.
\citet{betancourt2015} anatomised the funnel and showed that the optimal
parameterisation is governed by the relative contribution of prior and likelihood
to a group's posterior (centred when the likelihood dominates, non-centred when
the prior does) and
\citet{papaspiliopoulos2007} characterised the centred/non-centred trade-off in
the same terms.
\citet{betancourt2019incomplete} sharpened this into an equivalence: an
incomplete reparameterisation such as centering versus non-centering is
equivalent to modifying the Riemannian metric directly, so the two
parameterisations are mirror-image metric geometries rather than genuinely
different targets.
At the metric level, then, the obstruction is well understood.

What is not yet settled is whether there is structure \emph{beyond} the metric.
The base--fiber split of the joint space is a fiber bundle, and a metric on the
total space induces not only a Riemannian geometry but a \emph{connection}, a
rule for how the fiber must co-move with the base as the hyperparameters change.
A natural hypothesis, which we state here as our own, since to our knowledge
it does not appear in the statistics literature, is that the residual
obstruction is the \emph{curvature} of this connection, carried into the chain
as \emph{holonomy}: a net contraction of the fiber after the hyperparameters
traverse a closed loop (Figure~\ref{fig:geometry_intro}).
Were it true, this would be a connection-level account sitting above the
metric-level one, and it would tell the sampler something the metric does not:
that closed excursions in $(\mu,\sigma)$ leave a systematic, transportable
imprint on the group parameters.

This hypothesis belongs to the geometric program in MCMC, and it is that program
the proof is aimed at. Riemannian-manifold HMC \citep{girolami2011} and the
differential-geometric account of Hamiltonian flow
\citep{betancourt2017,livingstone2019} navigate hierarchical posteriors by
exploiting the curvature of the Fisher metric, while a parallel line learns
reparameterising transformations that absorb the funnel automatically
\citep{gorinova2020}. A connection-curvature account of centering would slot
directly into this program: it would licence samplers that
\emph{parallel-transport} the fiber as the hyperparameters move, rather than only
reparameterise, and would motivate proposals built from the connection. Whether
that structure exists is therefore worth settling before such a method is built.

Our central finding is that this hypothesis is false, and not just for the
canonical model.
The object is the \emph{metric-orthogonal Ehresmann connection}
$\Aconn = -\GFF^{-1}\GBF$, the horizontal distribution $G$-orthogonal to the
fibers, and we prove it is \emph{flat} for \emph{any} smooth hierarchical
posterior (Proposition~\ref{prop:universal_flat}): its horizontal leaves are the
level sets of the fiber score $\partial_{\bm\alpha}\log p$, so a global flat
trivialisation always exists and there is no holonomy to obstruct mixing.
Gaussianity, conjugacy, and the dimension of the fiber are irrelevant; the
logistic GLMM (Proposition~\ref{prop:flat}) is one closed-form instance.
The connection-level account collapses onto the metric-level one.
What survives is not geometric but statistical: the conditional dependence of
the fiber on the base, governed per group by the prior fraction $\pi_j$, the
pooling factor of \citet{gelman2006pooling}, the precise form of the
prior/likelihood balance whose role \citet{betancourt2015} established
geometrically, which the bundle geometry now produces in closed
form, and which, once the connection is known to be flat, is the \emph{only}
thing left governing the obstruction.
The Fisher information actually carries several connections, of which only this
one is at issue; we keep them apart in Section~\ref{sec:which_connection}.
The geometry is not thereby vacuous, however: genuine, even rotational,
curvature reappears as soon as one transports along a sampler's \emph{working}
metric rather than the true Hessian, and that is where holonomy returns --- as
an algorithmic rather than a geometric phenomenon
(Section~\ref{sec:working_metric}).

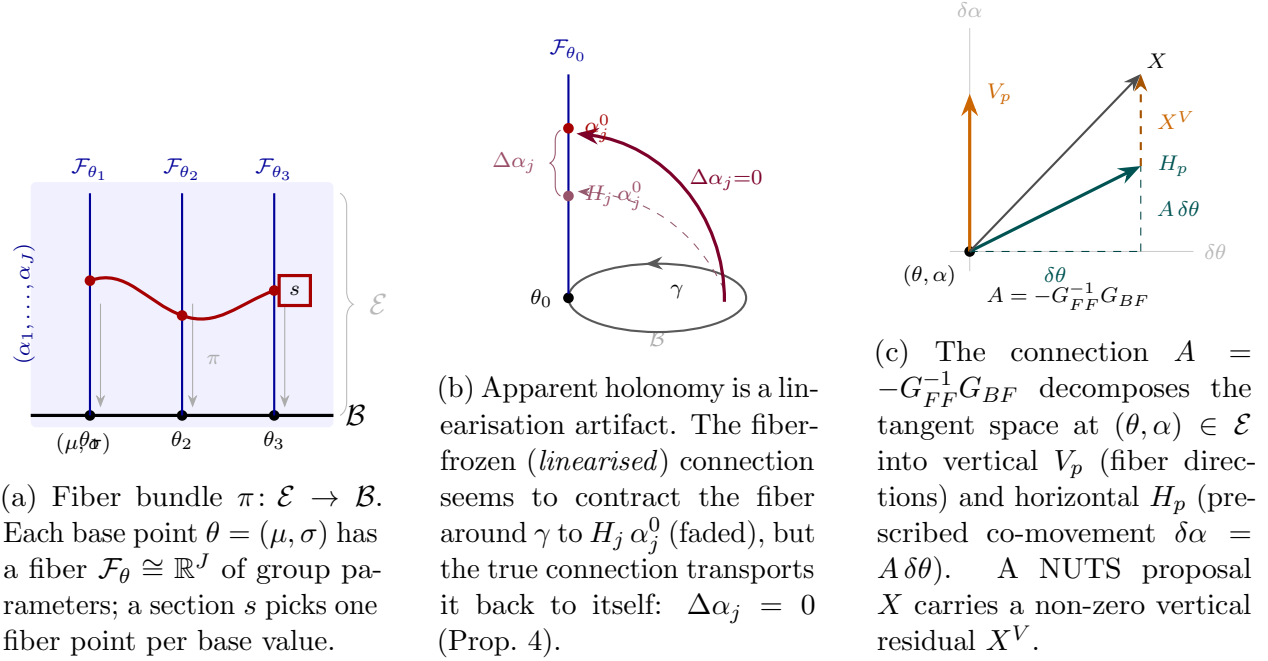
\begin{figure}[t]
  \centering
  \begin{subfigure}[b]{0.30\linewidth}
    \centering
    \begin{tikzpicture}[>=Stealth, scale=0.87,
        F/.style={draw=blue!60!black, thick},
        S/.style={draw=red!65!black, very thick},
        P/.style={->, gray!65, thin},
        D/.style={circle, fill, inner sep=0pt, minimum size=4pt},
      ]
      \fill[blue!6, rounded corners=2pt]
        (-2.3,-0.18) rectangle (2.3,3.55);
      \draw[black, very thick] (-2.3,0) -- (2.3,0)
        node[right, font=\small] {$\Bs$};
      \node[font=\scriptsize, below=2pt] at (-1.5,-0.02)
        {$(\mu,\sigma)$};
      \foreach \x/\i in {-1.4/1, 0/2, 1.4/3}{
        \node[D] at (\x,0) {};
        \node[below=2pt, font=\scriptsize] at (\x,0) {$\theta_\i$};
        \draw[F] (\x,0) -- (\x,3.38);
        \node[above=1pt, blue!60!black, font=\scriptsize]
          at (\x,3.38) {$\Fs_{\theta_\i}$};
        \draw[P] (\x+.16,1.7) -- (\x+.16,0.12);
      }
      \node[gray!65, font=\scriptsize] at (.48,.85) {$\pi$};
      \draw[decorate, decoration={brace, mirror, amplitude=5pt}, gray!55]
        (2.4,0) -- (2.4,3.38)
        node[midway, right=7pt, gray!55, font=\small] {$\Tot$};
      \node[rotate=90, blue!55!black, font=\scriptsize]
        at (-2.4,1.7) {$(\alpha_1,\ldots,\alpha_J)$};
      \draw[S] (-1.4,2.05)
        to[out=20,in=160] (0,1.52)
        to[out=-20,in=200] (1.4,1.90);
      \foreach \x/\h in {-1.4/2.05, 0/1.52, 1.4/1.90}
        \node[D, red!65!black] at (\x,\h) {};
      \node[S, right=1pt, font=\scriptsize] at (1.41,1.90) {$s$};
    \end{tikzpicture}
    \subcaption{Fiber bundle $\pi\colon\Tot\to\Bs$.  Each base
      point $\theta=(\mu,\sigma)$ has a fiber
      $\Fs_\theta\cong\R^J$ of group parameters; a section $s$
      picks one fiber point per base value.}
    \label{fig:intro_bundle}
  \end{subfigure}
  \hfill
  \begin{subfigure}[b]{0.30\linewidth}
    \centering
    \begin{tikzpicture}[>=Stealth, scale=0.87,
        L/.style={draw=black!65, thick,
          postaction={decorate,
            decoration={markings,
              mark=at position 0.28 with {\arrow{Stealth}}}}},
        F/.style={draw=blue!60!black, thick},
        H/.style={draw=purple!65!black, very thick},
        D/.style={circle, fill, inner sep=0pt, minimum size=4pt},
      ]
      \node[font=\scriptsize, gray!65] at (0,-0.62) {$\Bs$};
      \draw[L] (0,0) ellipse (1.35 and 0.52);
      \node[font=\scriptsize] at (0.3,0.09) {$\gamma$};
      \coordinate (T0) at (-1.35,0);
      \node[D] at (T0) {};
      \node[left=2pt, font=\scriptsize] at (T0) {$\theta_0$};
      \draw[F] ($(T0)+(0,0.06)$) -- ($(T0)+(0,3.4)$);
      \node[above=1pt, blue!60!black, font=\scriptsize]
        at ($(T0)+(0,3.4)$) {$\Fs_{\theta_0}$};
      \coordinate (A0) at ($(T0)+(0,2.58)$);
      \coordinate (HA0) at ($(T0)+(0,1.55)$);
      \node[D, purple!35!gray] at (HA0) {};
      \node[right=2pt, font=\scriptsize, purple!40!gray] at (HA0)
        {$H_j\,\alpha_j^0$};
      \draw[decorate, decoration={brace, amplitude=4pt}, purple!35!gray]
        ($(T0)+(-0.12,1.55)$) -- ($(T0)+(-0.12,2.55)$)
        node[midway, left=5pt, font=\scriptsize, purple!40!gray]
          {$\Delta\alpha_j$};
      \draw[->, dashed, purple!30!gray]
        (1.02,0.12) to[bend right=30] ($(T0)+(0.14,1.62)$);
      \draw[->, H]
        (1.02,-0.06) to[bend right=42]
        node[right=1pt, font=\scriptsize, purple!65!black, pos=0.55] {$\Delta\alpha_j{=}0$}
        ($(T0)+(0.12,2.5)$);
      \node[D, red!65!black] at (A0) {};
      \node[right=2pt, font=\scriptsize, red!65!black] at (A0)
        {$\alpha_j^0$};
    \end{tikzpicture}
    \subcaption{Apparent holonomy is a linearisation artifact.  The
      fiber-frozen (\emph{linearised}) connection seems to contract the fiber
      around $\gamma$ to $H_j\,\alpha_j^0$ (faded), but the true connection
      transports it back to itself: $\Delta\alpha_j=0$
      (Prop.~\ref{prop:flat}).}
    \label{fig:intro_holonomy}
  \end{subfigure}
  \hfill
  \begin{subfigure}[b]{0.30\linewidth}
    \centering
    \begin{tikzpicture}[>=Stealth, scale=0.87,
        V/.style={->, draw=orange!80!black, very thick},
        Hs/.style={->, draw=teal!65!black, very thick},
        X/.style={->, draw=black!70, thick},
        D/.style={circle, fill, inner sep=0pt, minimum size=4pt},
      ]
      \draw[gray!40, thin] (0,-0.3) -- (0,3.4)
        node[above, font=\scriptsize, gray!55] {$\delta\alpha$};
      \draw[gray!40, thin] (-0.3,0) -- (3.4,0)
        node[right, font=\scriptsize, gray!55] {$\delta\theta$};
      \node[D] at (0,0) {};
      \node[below left=1pt, font=\scriptsize] at (0,0)
        {$(\theta,\alpha)$};
      \draw[V] (0,0) -- (0,2.4)
        node[right=2pt, font=\scriptsize, orange!80!black] {$V_p$};
      \draw[Hs] (0,0) -- (2.6,1.3)
        node[right=2pt, font=\scriptsize, teal!65!black] {$H_p$};
      \draw[teal!55!black, thin, dashed]
        (2.6,0) -- (2.6,1.3)
        node[right=2pt, midway, font=\scriptsize, teal!65!black]
          {$\Aconn\,\delta\theta$};
      \draw[teal!55!black, thin, dashed]
        (0,0) -- (2.6,0)
        node[below=2pt, midway, font=\scriptsize, teal!65!black]
          {$\delta\theta$};
      \draw[X] (0,0) -- (2.6,2.7)
        node[above right=-2pt, font=\scriptsize] {$X$};
      \draw[->, draw=orange!65!black, thick, dashed]
        (2.6,1.3) -- (2.6,2.7)
        node[right=2pt, midway, font=\scriptsize, orange!80!black]
          {$X^V$};
      \node[font=\scriptsize, align=center] at (1.5,-0.68)
        {$\Aconn = {-}\GFF^{-1}\GBF$};
    \end{tikzpicture}
    \subcaption{The connection $\Aconn=-\GFF^{-1}\GBF$ decomposes
      the tangent space at $(\theta,\alpha)\in\Tot$ into vertical
      $V_p$ (fiber directions) and horizontal $H_p$ (prescribed
      co-movement $\delta\alpha=\Aconn\,\delta\theta$).  A NUTS
      proposal $X$ carries a non-zero vertical residual $X^V$.}
    \label{fig:intro_connection}
  \end{subfigure}
  \caption{Geometric structure of a hierarchical posterior.
    \emph{(a)}~The joint space $\Tot$ is a fiber bundle over the
    hyperparameter base $\Bs$; each fiber $\Fs_\theta\cong\R^J$
    carries the group parameters $(\alpha_1,\ldots,\alpha_J)$.
    \emph{(b)}~The naive holonomy picture, in which a closed loop
    $\gamma$ in $\Bs$ contracts the fiber by $H_j<1$.  This is the
    behaviour of the linearised connection only; we prove the true
    connection is flat (Prop.~\ref{prop:flat}), so its holonomy is
    trivial.
    \emph{(c)}~The Ehresmann connection $\Aconn=-\GFF^{-1}\GBF$
    defines a horizontal subspace $H_p$ at each point; a generic
    NUTS proposal $X$ has a non-zero vertical residual $X^V$
    because the sampler does not know $\Aconn$.}
  \label{fig:geometry_intro}
\end{figure}

Standard diagnostics do not identify this mechanism.
$\hat{R}$ \citep{gelman1992, vehtari2021} measures chain-to-chain variability.
ESS measures autocorrelation.
Divergence counts flag regions of high curvature in the Hamiltonian, not in the
parameter-space geometry.
All of these register the \emph{symptoms} of the obstruction (slow mixing in
the affected coordinates), but none localises the base--fiber coupling to
specific groups, and none indicates the remedy (which groups to
reparameterise).
The empirical sections show that the mechanism's chain-level signature is
structured conditional autocorrelation localised in prior-dominated groups
and predicted by the analytic prior fraction $\pi_j$, rather than
transport-following drift --- consistent with the connection being flat.

\subsection{Contributions}

This paper makes the following contributions.

\begin{enumerate}
  \item \textbf{Geometric framework and a universal flatness theorem}
    (Section~\ref{sec:geometry}).
    We formalise the fiber bundle structure of hierarchical posterior parameter
    spaces, on which the Fisher information metric induces a natural Ehresmann
    connection $A = -\GFF^{-1}\GBF$.
    Our main geometric result (Proposition~\ref{prop:universal_flat}) is that
    this connection is \emph{flat} for \emph{any} smooth hierarchical posterior,
    of any fiber dimension and any prior: its horizontal leaves are the level
    sets of the fiber score $\partial_{\bm\alpha}\log p$, so the holonomy is
    trivial and the closed-form logistic-GLMM case
    (Proposition~\ref{prop:flat}) is a corollary.
    The nonzero exterior-derivative curvature $F_j = -2/(\sigma^5 G_{FF,j}^2)$
    is only the fiber-frozen \emph{linearisation}; the true connection
    transports trivially, and the synthetic ``holonomy'' it appears to show is
    that linearisation artifact.

  \item \textbf{The prior fraction as the operative quantity}
    (Section~\ref{sec:glmm_connection}).
    With the connection flat, the obstruction reduces to the conditional
    dependence of fiber on base, governed per group by the prior fraction
    $\pi_j = (1/\sigma^2)/G_{FF,j}$.
    This quantity is not new: it is exactly the pooling factor $\omega$ of
    \citet{gelman2006pooling} and the per-group form of the prior/likelihood
    balance that \citet{betancourt2015} tied to the optimal parameterisation.
    Our contribution is to derive it from the bundle geometry as the
    \emph{only} quantity left governing the obstruction once the connection is
    known to be flat, to give it in closed form for the GLMM, and to compute it
    per coordinate from a fitted model.

  \item \textbf{A loop-conditional dependence diagnostic and an attribution
    null} (Section~\ref{sec:diagnostic}).
    We propose a chain-based diagnostic that estimates per-group
    loop-conditional dependence coefficients $\hat{\rho}_j$ by weighted least
    squares over approximate loops in the hyperparameter chain.
    A direct attribution test, comparing observed fiber displacements with
    parallel-transport predictions along the same chain segments, returns a
    null result, exactly as flatness requires: NUTS does not transport the
    fiber, so the chain-level signature is excess conditional autocorrelation
    localised in prior-dominated groups, not transport-following drift.

  \item \textbf{Where curvature genuinely lives: working-metric connections}
    (Section~\ref{sec:working_metric}).
    Flatness holds for the connection of the \emph{true} Hessian metric.
    For a connection built from a \emph{working} metric (a sampler's fixed
    mass matrix), the curvature is generically nonzero and can be rotational,
    which we demonstrate with derivatives validated against automatic
    differentiation.
    A fixed-metric sampler transports along this curved connection, so holonomy
    re-enters as an \emph{algorithmic} rather than geometric phenomenon --- the
    natural bridge to connection-aware sampling.

  \item \textbf{Simulation study} (Section~\ref{sec:simulation}).
    A $4 \times 4$ factorial study over observations per group
    ($n_j \in \{3,10,30,100\}$) and hierarchical standard deviation
    ($\sigma_\text{true} \in \{0.5,1,2,3\}$), 10 replicates per cell, confirms
    that the diagnostic signal tracks the analytic prior fraction.
    The relationship with $\sigma$ is non-monotone, separating funnel severity
    (a base-space pathology) from the base--fiber coupling that $\pi_j$
    measures.

  \item \textbf{Practitioner guidance and a controlled comparison}
    (Section~\ref{sec:guidance}).
    We situate $\pi_j$ within established remedies, per-group partial
    non-centering \citep{papaspiliopoulos2003, gorinova2020}, interweaving
    \citep{yu2011}, and fiber marginalisation, so a flagged group comes with
    a standard-tool response, with $w_j \approx \pi_j$ supplying the per-group
    weight analytically.
    A controlled mixed-design comparison shows this rule is a safe interpolator
    that tracks the better uniform parameterisation without surpassing it when
    the min-ESS bottleneck is the $\sigma$-funnel.
    A connection-corrected block sampler (\texttt{horizontal\_mcmc}) is included
    as a proof of concept.

  \item \textbf{The \texttt{fibr} R package} (Section~\ref{sec:package}).
    The prior fraction (with a \texttt{brms} adapter) and the
    \texttt{smoothbp\_advisor()} companion are distributed as the open-source R
    package \texttt{fibr} \citep{bindoff2026fibr}; the connection, dependence, and
    sampler methods accompany the paper as reproduction code in the repository.
\end{enumerate}

\subsection{Notation}

Scalars are unadorned ($\mu, \sigma, \alpha_j$).
Vectors are lower-case bold ($\bm{\alpha}$).
Matrices are upper-case ($G, A, H$).
We write $J$ for the number of groups, $K$ for the dimension of the base space
($K=2$ throughout: $\theta = (\mu, \sigma)$), $M$ for the number of detected
loop pairs, and $N$ for the total number of observations.
The notation $[J]$ denotes $\{1, \ldots, J\}$.
$\norm{\cdot}_F$ is the Frobenius norm.

\section{Background}
\label{sec:background}

\subsection{Hierarchical models and the centering problem}

The two-level normal hierarchy, and its generalisations to GLMMs, are the
canonical setting where posterior geometry creates sampling difficulty.
\citet{papaspiliopoulos2003, papaspiliopoulos2007} characterised the
centering/non-centering trade-off: the centred parameterisation is efficient when
the likelihood identifies the group-level parameters well, while the non-centred
parameterisation is efficient when the prior dominates.
\citet{betancourt2015} described the geometric anatomy of the resulting funnel and
established the role of the prior/likelihood balance in determining the optimal
parameterisation.
\citet{betancourt2019incomplete} gave the sharpest geometric statement to date:
centering versus non-centering is an \emph{incomplete} reparameterisation,
equivalent to a change of the Riemannian metric, so the two parameterisations
are mirror-image metric geometries.
These are metric-level accounts, and they are essentially complete at that
level; what they leave open, and what we take up, is whether the
base--fiber \emph{connection} adds anything beyond the metric.
The per-group quantity at the centre of our answer, the prior fraction $\pi_j$,
makes this balance precise and coincides with the shrinkage / pooling
factor of \citet{gelman2006pooling}; Sections~\ref{sec:geometry}
and~\ref{sec:glmm_connection} derive it from the bundle geometry.

\subsection{Riemannian manifold HMC}
\label{sec:rmhmc_bg}

\citet{girolami2011} proposed Riemannian manifold Langevin and Hamiltonian Monte
Carlo (RMHMC), using the Fisher information metric as a position-dependent mass
matrix.
This adapts the sampler to the local curvature of the posterior manifold and
dramatically improves performance in the funnel and related geometries.
\citet{betancourt2013} extended this to a general metric framework.
\citet{zhang2014} proposed a semi-separable decomposition for hierarchical models,
splitting the Hamiltonian into base and fiber contributions with a Schur complement
correction.
\citet{kleppe2024} showed that the log-density gradient covariance provides a
practical proxy for the Fisher metric without requiring analytical computation.

These are all metric-adaptation methods: they precondition with the Fisher
metric or a proxy for it.
Our results bear on them in two opposite ways.
First, the base--fiber Schur-complement correction of \citet{zhang2014} is built
from the same blocks as our connection $\Aconn = -\GFF^{-1}\GBF$, which is the
horizontal lift of a base move; Proposition~\ref{prop:universal_flat} shows this
connection is flat, so for the \emph{true} Fisher metric there is no horizontal
information beyond the metric to exploit --- consistent with the
reparameterisation--metric equivalence of \citet{betancourt2019incomplete}.
Second, and conversely, a sampler that \emph{fixes} its mass matrix, as adaptive
HMC does after warmup, transports the fiber along the connection of that
\emph{working} metric, which is generically curved
(Section~\ref{sec:working_metric}).
That is the regime in which a connection-aware correction has genuine content,
and it is the subject of a companion paper.

\subsection{Fiber bundles in statistics}

Differential and information geometry have a long history in statistics:
the dually-flat structure of exponential families and the
$\alpha$-connections of \citet{amari2000} are standard, and tangent-bundle
constructions appear throughout that literature.
What is new here, to our knowledge, is the specific object: the
metric-orthogonal \emph{Ehresmann connection} of the base--fiber split of a
hierarchical posterior, $A=-\GFF^{-1}\GBF$, together with its curvature
analysis and the resulting flatness theorem.
The connection describes how fibers are parallel-transported as the base
moves.
We give a self-contained treatment of the necessary geometry in
Section~\ref{sec:geometry}; accessible references include
\citet{kobayashi1963} and \citet{nakahara2003}.

\section{The Geometric Framework}
\label{sec:geometry}

\subsection{Fiber bundle structure of hierarchical posteriors}

Consider a two-level hierarchical model with hyperparameters
$\bm{\theta} \in \Bs \subseteq \R^K$ (the \emph{base space}) and
group-level parameters $\bm{\alpha} \in \R^J$ (the \emph{fiber}).
The joint parameter space is the \emph{total space}
$\Tot = \Bs \times \R^J$, equipped with the projection
$\pi: \Tot \to \Bs$, $\pi(\bm{\theta}, \bm{\alpha}) = \bm{\theta}$.

The fiber over a point $\bm{\theta} \in \Bs$ is
$\Fs_{\bm{\theta}} = \pi^{-1}(\bm{\theta}) \cong \R^J$,
representing the group-level parameters conditioned on the hyperparameters.
In the canonical GLMM:
\begin{align}
  \Bs &= \{(\mu, \sigma) : \mu \in \R,\, \sigma > 0\}, \quad K = 2, \\
  \Fs_{(\mu,\sigma)} &= \R^J \quad
    \text{with prior } \alpha_j \mid \mu, \sigma \sim \Normal(\mu, \sigma^2).
\end{align}
Fixed effects $\bm{\beta}$ may be treated as part of the base space or
marginalised; for clarity we include them in the base throughout and write
$\bm{\theta} = (\mu, \sigma, \bm{\beta})$ when they are relevant.

\begin{remark}
  The bundle $\Tot \to \Bs$ is topologically trivial ($\Tot \cong \Bs \times \R^J$).
  The off-diagonal Fisher metric blocks $\GBF$ couple base and fiber in a way
  that depends on the current parameter values, so the induced connection is
  non-trivial as a \emph{distribution}; whether it is \emph{curved} is a
  separate question, and Section~\ref{sec:glmm_connection} shows that for the
  GLMM it is not (the connection is flat).
  The centering/non-centering debate is, geometrically, a search for a
  trivialisation that flattens the induced connection; flatness guarantees one
  exists.
\end{remark}

\subsection{The Fisher metric and its block structure}

The Fisher information metric on $\Tot$ is the expected negative Hessian of the
log-posterior:
\begin{equation}
  G_{pq}(\bm{\theta}, \bm{\alpha}) =
  -\E_{y \mid \bm{\theta}, \bm{\alpha}}\!\left[
    \frac{\partial^2 \log p(y, \bm{\theta}, \bm{\alpha})}
         {\partial \xi_p \, \partial \xi_q}
  \right],
  \label{eq:fisher_metric}
\end{equation}
where $\bm{\xi} = (\bm{\theta}, \bm{\alpha})$ concatenates all parameters.
The block structure of $G$ in the base--fiber decomposition is
\begin{equation}
  G =
  \begin{pmatrix} \GBB & \GBF^\top \\ \GBF & \GFF \end{pmatrix},
  \label{eq:G_block}
\end{equation}
where $\GFF \in \R^{J \times J}$ is the fiber metric block,
$\GBF \in \R^{J \times K}$ is the off-diagonal coupling block, and
$\GBB \in \R^{K \times K}$ is the base metric block.
The off-diagonal block $\GBF$ encodes how base and fiber directions are coupled:
it is non-zero whenever the prior creates statistical dependence between
$\bm{\theta}$ and $\bm{\alpha}$.

Two points fix the precise object the connection is built from.
First, $\GBF$ arises only from the prior $p(\bm{\alpha}\mid\bm{\theta})$,
the likelihood does not depend on the hyperparameters, so it carries no
expectation over $y$:
$\GBF = -\partial_{\bm{\alpha}}\partial_{\bm{\theta}}\log p(\bm{\alpha}\mid\bm{\theta})$.
Second, the connection \eqref{eq:connection} and the flatness analysis below use
the negative Hessian of the \emph{realised} log-posterior, the observed
information the sampler sees at a point.
This coincides with the expected Fisher metric \eqref{eq:fisher_metric}
whenever the likelihood Hessian is data-independent, as it is for canonical-link
GLMs, including the logistic GLMM of Section~\ref{sec:glmm_connection}, where
$-\partial^2_{\alpha_j}\log p(y\mid\alpha_j) = \sum_i p_{ji}(1-p_{ji})$ does not
depend on $y$.
We keep the name \emph{Fisher information metric} for both, in keeping with the
RMHMC literature; the distinction matters only for non-canonical links, where
the observed-information connection is the relevant object.

\subsection{The Ehresmann connection}

\begin{definition}[Ehresmann connection]
  Given the block metric \eqref{eq:G_block}, the \emph{Ehresmann connection
  form} is the $J \times K$ matrix-valued function
  \begin{equation}
    \Aconn(\bm{\theta}, \bm{\alpha})
    = -\GFF(\bm{\theta}, \bm{\alpha})^{-1}\,
       \GBF(\bm{\theta}, \bm{\alpha}).
    \label{eq:connection}
  \end{equation}
  The \emph{horizontal lift} of a base tangent vector
  $\bm{v}_B \in T_{\bm{\theta}}\Bs$ to a point
  $(\bm{\theta}, \bm{\alpha}) \in \Tot$ is the total-space vector
  $\tilde{\bm{v}} = (\bm{v}_B,\, \Aconn \bm{v}_B)$: the unique direction
  in $T\Tot$ that is $G$-orthogonal to the fiber.
\end{definition}

Geometrically, $\Aconn$ tells us how the fiber must co-move with the base to
remain ``stationary'' relative to the posterior geometry.
Moving the hyperparameters by $\delta\bm{\theta}$ while keeping the fiber
on the horizontal slice requires a simultaneous fiber displacement of
$\delta\bm{\alpha} = \Aconn\,\delta\bm{\theta}$.

\begin{remark}[Relationship to centering]
  In the centred parameterisation of a normal--normal hierarchy,
  $A_{j,\mu} = 1/(\sigma^2 G_{FF,j})$ (derived in Section~\ref{sec:glmm_connection}).
  This is positive and large when $\sigma$ is small (tight prior, high prior
  precision), meaning the conditional posterior of every $\alpha_j$ tracks
  $\mu$ tightly.
  The centred chain explores this slowly; we show below that the obstruction is
  this conditional dependence rather than any holonomy of $\Aconn$.
\end{remark}

\subsection{Which connection? A disambiguation}
\label{sec:which_connection}

The Fisher information carries several connections, and our results concern one
of them specifically.
We separate three objects on which the words ``flat'' and ``curved'' have
different answers.

\emph{(i) The Levi-Civita connection of the Fisher--Rao metric on the base.}
The intrinsic Riemannian geometry of the location-scale family
$(\mu, \sigma)$ is hyperbolic: the Fisher--Rao metric has constant negative
curvature and is \emph{not} flat \citep{amari2000}.
This is the genuine curvature of the base manifold, and it is the geometric
content of the funnel.
Nothing in this paper removes or trivialises it.

\emph{(ii) The Amari $e$- and $m$-connections.}
Because the Gaussian is an exponential family, the dual $\alpha = \pm 1$
connections are flat: affine (natural and expectation) coordinates exist in
which they vanish \citep{amari2000}.
That the $\alpha = \pm 1$ connections are flat while the $\alpha = 0$
(Levi-Civita) connection of (i) is not is the standard fact that dual flatness
does not imply Riemannian flatness.

\emph{(iii) The metric-orthogonal Ehresmann connection $\Aconn$ of
\eqref{eq:connection}.}
This is not an affine connection on a statistical manifold but an Ehresmann
connection on the parameter bundle $\Tot \to \Bs$: a horizontal distribution,
$G$-orthogonal to the fibers.
Propositions~\ref{prop:universal_flat} and~\ref{prop:flat} concern this object,
and \emph{only} this object: its curvature, the obstruction to integrability of
the horizontal distribution, vanishes identically.

Every ``flat'', ``curvature'', and ``holonomy'' statement in
Sections~\ref{sec:diagnostic}--\ref{sec:discussion} refers to $\Aconn$ of
(iii).
The nonzero $F_j$ of \eqref{eq:curvature} below is the curvature of the
\emph{linearised} connection $\tilde\Aconn(\bm\theta) := \Aconn(\bm\theta,
\bm\alpha_0)$ obtained by freezing the fiber, a different Ehresmann
connection that approximates $\Aconn$ near a fixed fiber point, and is the
quantity our synthetic-loop validation integrates.
When we later call this nonzero $F_j$ an ``artifact,'' we mean precisely that it
is the curvature of $\tilde\Aconn$, not of $\Aconn$; we do \emph{not} mean the
Fisher--Rao curvature of (i), which is real and which we leave untouched.
Why $\Aconn$ is flat, its horizontal leaves are the level sets of the fiber
score, is the subject of Section~\ref{sec:universal_flat}.
The flat coordinate that results is the conditional score, which for the GLMM is
the \emph{precision}-weighted residual $(\alpha_j-\mu)/\sigma^2$, not the
SD-weighted non-centering $(\alpha_j-\mu)/\sigma$; both reparameterise the fiber,
but only the score is horizontal for $\Aconn$.

\subsection{Holonomy and curvature}

\begin{definition}[Holonomy]
  Let $\gamma: [0,1] \to \Bs$ be a closed loop in the base
  ($\gamma(0) = \gamma(1) = \bm{\theta}_0$).
  Starting from a fiber point $\bm{\alpha}_0 \in \Fs_{\bm{\theta}_0}$, the
  \emph{parallel transport} of $\bm{\alpha}_0$ along $\gamma$ is the path
  $t \mapsto \bm{\alpha}(t)$ in $\Tot$ satisfying
  \begin{equation}
    \dot{\bm{\alpha}}(t) = \Aconn(\gamma(t), \bm{\alpha}(t))\,\dot{\gamma}(t),
    \quad \bm{\alpha}(0) = \bm{\alpha}_0,
    \label{eq:transport}
  \end{equation}
  and the \emph{holonomy} is the net displacement
  $\Delta\bm{\alpha} = \bm{\alpha}(1) - \bm{\alpha}_0$ after traversing the loop.
\end{definition}

The \emph{curvature} of the connection is the obstruction to integrability of
its horizontal distribution, equivalently, the vertical component of the Lie
bracket of two horizontal lifts, and it governs holonomy: for a small loop
enclosing area $\mathcal{A}$, the holonomy is
$\Delta\bm{\alpha} \approx (\text{curvature}) \cdot \mathcal{A} +
O(\mathcal{A}^{3/2})$, so a connection of zero curvature, a \emph{flat}
connection, has trivial holonomy around contractible loops.
The entire centring/non-centring intuition rests on the hypothesis that
$\Aconn$ is curved.
The next subsection shows it is not.

\subsection{Universal flatness: the score is a first integral}
\label{sec:universal_flat}

The curvature in \eqref{eq:curvature}--\eqref{eq:stokes} is the linearised one.
The \emph{full} curvature of $\Aconn$ is governed by a single fact that does not
depend on the model at all.

\begin{proposition}[Universal flatness of the metric connection]
\label{prop:universal_flat}
  Let $\ell(\bm{\theta}, \bm{\alpha})$ be a smooth log-posterior on
  $\Tot = \Bs \times \R^J$ with fiber-precision block
  $\GFF = -\partial^2_{\bm{\alpha}}\ell$ nonsingular on an open region
  $U \subseteq \Tot$, and $\GBF = -\partial_{\bm{\alpha}}\partial_{\bm{\theta}}\ell$.
  The metric-orthogonal Ehresmann connection $\Aconn = -\GFF^{-1}\GBF$ is
  \emph{flat} on $U$: its horizontal distribution is integrable, with the
  \emph{fiber score} $\bm{s}(\bm{\theta},\bm{\alpha}) = \partial_{\bm{\alpha}}\ell$
  as a (vector) first integral.
  On any simply connected subregion the holonomy is trivial; and where the
  conditional score $\bm{\alpha} \mapsto \partial_{\bm{\alpha}}\ell$ is a
  diffeomorphism (in particular when the conditional posterior is log-concave,
  as in the logistic GLMM), the map
  $(\bm{\theta}, \bm{\alpha}) \mapsto (\bm{\theta}, \bm{s})$ is a global flat
  trivialisation; otherwise it is local.
\end{proposition}

\begin{proof}
  The horizontal lift of $\bm{v} \in T_{\bm{\theta}}\Bs$ is
  $(\bm{v}, \Aconn\bm{v})$.
  Differentiating the score along it,
  \[
    d\bm{s}(\bm{v}, \Aconn\bm{v})
    = \partial^2_{\bm{\alpha}\bm{\alpha}}\ell\,(\Aconn\bm{v})
      + \partial^2_{\bm{\alpha}\bm{\theta}}\ell\,\bm{v}
    = (-\GFF)(-\GFF^{-1}\GBF\bm{v}) + (-\GBF)\bm{v}
    = \bm{0}.
  \]
  So $\bm{s}$ is constant along horizontal curves and the horizontal
  distribution lies in $\ker d\bm{s}$.
  Where $\GFF$ is nonsingular, $d\bm{s}$ has rank $J$ (its fiber block $-\GFF$ is
  invertible), so $\ker d\bm{s}$ has dimension $K$, equal to the rank of the
  horizontal distribution; the two coincide.
  The horizontal distribution is therefore the tangent distribution of the
  foliation of $U$ by the level sets of the globally defined map $\bm{s}$, hence
  integrable, and its curvature vanishes.
\end{proof}

Three consequences are worth stating.
First, \emph{the flatness has nothing to do with conjugacy, Gaussianity, or the
dimension of the fiber}: the proof uses only that $\ell$ is smooth and $\GFF$ is
nonsingular.
The closed-form GLMM computation of Section~\ref{sec:glmm_connection}
(Proposition~\ref{prop:flat}) is one instance, retained because the explicit
cancellation is instructive and because it pins down the linearised curvature
$F_j$ (Section~\ref{sec:linearised_curvature}).
Second, the flat coordinate is the conditional score
$\bm{s} = \partial_{\bm{\alpha}}\ell$, equivalently the residual to the
conditional mode (where $\bm{s} = \bm{0}$); for a Gaussian prior this is the
precision-weighted residual $(\alpha_j - \mu)/\sigma^2$.
Third, the hypotheses are sharp.
Flatness can fail only where $\GFF$ is \emph{singular} (a degenerate or
unidentified conditional, the fiber-degeneracy mode of
Section~\ref{sec:discussion}) or for a connection built from a metric other
than the true Hessian.
The latter is where genuine, and even rotational, curvature appears; we return
to it in Section~\ref{sec:working_metric}.

\subsection{The linearised curvature}
\label{sec:linearised_curvature}

A non-zero curvature reappears the moment one \emph{freezes} the fiber.
Treating $\Aconn$ as a function of the base alone, holding $\bm{\alpha}$ at a
fixed $\bm{\alpha}_0$, its exterior derivative gives, for the $K = 2$ base, a
$J$-vector
\begin{equation}
  F_j = \frac{\partial A_{j,2}}{\partial \theta^1}
        - \frac{\partial A_{j,1}}{\partial \theta^2},
  \quad j \in [J],
  \label{eq:curvature}
\end{equation}
where $A_{j,k}$ is the $(j,k)$ entry of $\Aconn$ (for a base of dimension $K$ the
curvature is a $J \times K \times K$ antisymmetric tensor; only $K = 2$ is
needed here).
This is the curvature of the \emph{linearised} (fiber-frozen) connection
$\tilde\Aconn(\bm{\theta}) = \Aconn(\bm{\theta}, \bm{\alpha}_0)$.
The full curvature of $\Aconn$ adds the vertical terms
$\Aconn\,\partial_{\bm{\alpha}}\Aconn$ that
Proposition~\ref{prop:universal_flat} cancels exactly; $F_j$ keeps only the first
piece.
Its Stokes prediction,
\begin{equation}
  \Delta\alpha_j \approx F_j(\bm{\theta}_0,\bm{\alpha}_0)\,\mathcal{A}
  + O(\mathcal{A}^{3/2}),
  \label{eq:stokes}
\end{equation}
is therefore the holonomy of $\tilde\Aconn$, not of $\Aconn$, whose holonomy is
zero.
The linearised object is not a fiction: it is what one obtains by evaluating the
connection at a fixed fiber point and ignoring its motion, which is exactly what
a numerical validation that holds $\bm{\alpha}_0$ fixed computes
(Section~\ref{sec:stokes_validation}) and what the chain diagnostic sees at the
loop time-scale.
We compute $F_j$ in closed form for the GLMM in
Section~\ref{sec:glmm_connection}.

\subsection{What the framework requires of a model}

To apply the dependence diagnostic (Section~\ref{sec:diagnostic}) to a new
hierarchical model, the practitioner needs only a posterior chain:
the diagnostic algorithm uses the chain draws themselves and requires no
analytical computation.
To compute the analytic connection and curvature
(Section~\ref{sec:glmm_connection}), one additionally needs closed-form or
automatic-differentiation-computed expressions for $\GFF$ and $\GBF$.
For the two-level logistic GLMM, both are available analytically; the
extension to models where AD must be used is future work
(Section~\ref{sec:discussion}).

\section{The Loop-Conditional Dependence Diagnostic}
\label{sec:diagnostic}

Because the connection is flat (Section~\ref{sec:glmm_connection}), the
chain-level quantity this diagnostic estimates is not geometric holonomy but
\emph{loop-conditional fiber dependence}: the residual dependence of the fiber
at the end of an approximate base-space loop on its value at the start.
We denote the per-group coefficient $\hat{\rho}_j$ to distinguish it from the
(trivial) geometric holonomy $H_j$ of Section~\ref{sec:glmm_connection}.
The diagnostic is implemented as \texttt{holonomy\_diagnostic()} in
\texttt{fibr}; the function name is retained for continuity.

Why estimate this at all, given that Proposition~\ref{prop:universal_flat}
already settles the geometry and the analytic prior fraction $\pi_j$
(Section~\ref{sec:glmm_connection}) localises the coupling more cleanly?
The diagnostic is \emph{secondary} to $\pi_j$ and we treat it as such, but it
earns a place on three grounds.
First, it is \emph{model-agnostic}: it uses only the posterior draws and needs
no closed-form or differentiated metric blocks, so it applies to models where
$\GFF$ and $\GBF$ are unavailable and $\pi_j$ cannot be computed.
Second, it is \emph{loop-conditional} --- it measures fiber dependence
conditioned on the hyperparameters returning to (approximately) where they
started, which is not the same as a plain per-group effective sample size: that
conditioning is what isolates the base--fiber coupling from generic fiber
autocorrelation.
Third, it is the \emph{empirical} counterpart of $\pi_j$, and the agreement
between the two (Section~\ref{sec:simulation}) is what ties the analytic theory
to chain behaviour.
Where $\pi_j$ is available it is the primary, more robust per-group flag
(Section~\ref{sec:discussion}); the diagnostic corroborates and extends it.

\subsection{Overview}

Given a posterior chain on $(\bm{\theta}, \bm{\alpha})$, the diagnostic
proceeds in four steps:
(1) detect approximate loops in the base space $\Bs$,
(2) estimate per-group loop-conditional dependence coefficients $\rho_j$ from
the fiber values at loop endpoints,
(3) assess sensitivity to the minimum loop gap $g$, and
(4) corroborate flatness empirically (the attribution test) by comparison with
the analytically integrated parallel transport
(Section~\ref{sec:attribution}).

A point of interpretation governs everything that follows, so we state it
up front.
For a chain in equilibrium, the fiber value at the end of a detected loop is
a draw from the conditional posterior given (approximately) the same base
point as at the start.
If the chain mixed instantly, $\bm{\alpha}_e$ and $\bm{\alpha}_s$ would be
independent draws from that conditional, and the best linear dependence map
would be $P \approx 0$ (after residualisation; Section~\ref{sec:transport_estimation}).
A slowly mixing chain instead retains memory of $\bm{\alpha}_s$, giving
$P$ closer to the identity.
The raw coefficients therefore measure \emph{loop-conditional fiber
dependence at the loop time-scale}.
In principle this could be driven by connection geometry (systematic, signed,
area-dependent transport) as well as by generic autocorrelation (unsigned
persistence); the flatness result rules out the former, but we retain the
attribution step in Section~\ref{sec:attribution}, which tests directly
whether the \emph{direction and magnitude} of fiber displacement across each
loop matches the analytic parallel-transport prediction, and which returns the
null result that flatness predicts.

\subsection{Loop detection}
\label{sec:loop_detection}

A \emph{loop pair} is a pair of chain iterations $(s, e)$ with $e - s \geq g$
(minimum gap $g$) such that the base-space distance
$\norm{\bm{\theta}_e - \bm{\theta}_s}_2$ is small.
The gap $g$ ensures that loop endpoints are not simply adjacent
autocorrelated states.

\begin{algorithm}
\caption{Loop detection in base space}
\label{alg:loop_detection}
\begin{algorithmic}[1]
\Require Base-space chain $\{\bm{\theta}_t\}_{t=1}^T$, tolerance $\varepsilon$,
         minimum gap $g$, number of neighbours $k$, maximum loops $L_{\max}$
\State Standardise columns of $\{\bm{\theta}_t\}$ to unit variance.
\State If $\varepsilon = \texttt{NULL}$: set $\varepsilon$ to the 5th percentile
       of a sub-sampled pairwise distance distribution.
\State Build a $k$-nearest-neighbour index over $\{\bm{\theta}_t\}$.
\State \textbf{for} each $t = 1, \ldots, T$:
\State \quad Query the $k$ nearest neighbours of $\bm{\theta}_t$; let $\mathcal{N}_t$
       be those with distance $< \varepsilon$ and index gap $> g$.
\State \quad Record pairs $(t, s)$ for each $s \in \mathcal{N}_t$.
\State Keep the $L_{\max}$ tightest pairs (smallest $\norm{\bm{\theta}_e - \bm{\theta}_s}_2$).
\Ensure Data frame of loop pairs $(s, e, d)$ where $d = \norm{\bm{\theta}_e - \bm{\theta}_s}_2$.
\end{algorithmic}
\end{algorithm}

Loop detection is run \emph{per chain} so that independent chains are never
spliced.
Pairs from all chains are pooled for the subsequent regression.

\subsection{Transport map estimation}
\label{sec:transport_estimation}

By Remark~\ref{rem:abelian}, the transport map for the GLMM class is
diagonal with real entries: parallel transport scales each fiber coordinate
independently.
We therefore estimate $J$ per-group scalars rather than a full
$J \times J$ matrix.
Given $M$ loop pairs with fiber start values $\alpha_{s_k,j}$ and end values
$\alpha_{e_k,j}$, the weighted least squares estimate for group $j$ is
\begin{equation}
  \hat{\rho}_j
  = \frac{\sum_{k=1}^{M} w_k\, \alpha_{e_k,j}\, \alpha_{s_k,j}}
         {\sum_{k=1}^{M} w_k\, \alpha_{s_k,j}^2 + \delta},
  \label{eq:wls_H}
\end{equation}
where $w_k$ are loop weights and $\delta > 0$ is a small ridge penalty.
We write $\hat{P} = \diag(\hat{\rho}_1, \ldots, \hat{\rho}_J)$.
The \texttt{fibr} implementation retains the unrestricted $J \times J$
estimator
$\hat{P} = (E W S^\top)(S W S^\top + \delta I)^{-1}$
as an option (\texttt{structure = "full"}) for models with genuinely coupled
fibers, but for diagonal-transport models the full estimator fits $J^2$
parameters to $M$ noisy pairs and manufactures spurious off-diagonal
structure, including spurious complex eigenvalues that can be mistaken
for rotational holonomy.
All results in this paper use the diagonal estimator.

Weights are based on loop tightness: $w_k \propto \exp(-d_k / \bar{d})$,
where $d_k = \norm{\bm{\theta}_{e_k} - \bm{\theta}_{s_k}}_2$ and
$\bar{d}$ is the mean distance, so that tighter loops (better approximations to
true closed loops) contribute more.

\paragraph{Residualisation.}
Before estimating $\hat{P}$, the fiber draws are residualised against the base
draws via OLS, removing the linear base-to-fiber effect
(e.g., $\alpha_j \approx \mu$ in the centred GLMM).
This isolates the vertical fiber component where geometric holonomy lives.

\paragraph{Interpretation.}
The model $\bm{\alpha}_e \approx P \bm{\alpha}_s$ is the simplest linear
summary of loop-conditional dependence.
For a well-mixed chain the endpoints are independent draws from the same
conditional, so $\hat{\rho}_j \approx 0$; a slowly mixing chain gives
$\hat{\rho}_j$ closer to unity.
Note that geometric holonomy would also have produced $\hat{\rho}_j \neq 0$,
but with a loop-orientation- and area-dependent signature; since the
connection is flat, no such geometric contribution exists, and any nonzero
$\hat{\rho}_j$ is conditional autocorrelation.
This is what the attribution test of Section~\ref{sec:attribution} confirms.

\subsection{Per-group dependence coefficients and scalar summary}

The estimated coefficients $\hat{\rho}_1, \ldots, \hat{\rho}_J$ are read as
follows:
\begin{itemize}
  \item $\hat{\rho}_j \approx 0$: loop endpoints are approximately independent
    draws from the conditional posterior; the chain mixes faster than the
    loop time-scale and no signal survives.
  \item $\hat{\rho}_j$ substantially above 0: the fiber retains loop-conditional
    dependence at the loop time-scale.  Under flatness this is conditional
    autocorrelation; the attribution step
    (Section~\ref{sec:attribution}) and the gap profile
    (Section~\ref{sec:gap_sensitivity}) confirm there is no additional
    geometric component.
  \item $\hat{\rho}_j > 1$: indicates numerical instability or insufficient
    residualisation rather than a genuine effect.
\end{itemize}

The mean coefficient $\bar{\rho} = J^{-1}\sum_j \hat{\rho}_j$ is
our primary scalar summary, ranging from 0 (independent endpoints) toward 1
(full persistence).
The Frobenius deviation $\norm{\hat{P} - I}_F$ is also reported but is a less
discriminating summary: it equals $\sqrt{J}$ in both the $P \approx 0$ and
$P \approx I$ limits.
Because the coefficients are group-aligned, $\hat{\rho}_j$ can be compared
directly with per-group analytic quantities: the prior fraction $\pi_j$
\eqref{eq:prior_fraction} and the integrated transport prediction
(Section~\ref{sec:attribution}).

\subsection{Bootstrap uncertainty}

Uncertainty is quantified by resampling loop pairs with replacement.
For each of $B$ bootstrap replicates, equation~\eqref{eq:wls_H} is re-solved
on the resampled pairs, and the resulting factors are recorded.
Because the diagonal estimator preserves group identity, the bootstrap
distribution of each $\hat{\rho}_j$ is directly interpretable as per-group
uncertainty, displayed as intervals on the per-group transport plot.

\subsection{Gap sensitivity}
\label{sec:gap_sensitivity}

The minimum gap $g$ sets the time-scale at which loops are detected.
If $g$ is inside the integrated autocorrelation time (IACT) of the fiber
chain, $\hat{\rho}_j$ is dominated by ordinary autocorrelation; as $g$ grows
past the IACT, the autocorrelation contribution decays.
Had the connection been curved, a systematic transport contribution would
re-accumulate on every loop regardless of elapsed time and persist beyond the
IACT; the gap profile was designed to expose such a component.
Flatness implies there is none, and the profiles bear this out.
We report $\hat{\rho}_j$ as a profile over a grid of gaps
(e.g.\ $g \in \{3, 10, 25, 50\}$) rather than at a single value, alongside
the estimated IACT of the base and fiber chains.
Two components can contribute to a signal surviving at $g$ beyond the IACT:
ordinary autocorrelation (decays with gap) and residualisation leakage
(gap-independent, arising because OLS removes only the \emph{linear}
base-to-fiber dependence, while loop endpoints share the same base point and
hence the same nonlinear conditional mean).
A negative-control experiment reported in Section~\ref{sec:attribution}
separates these contributors empirically.

\subsection{Attribution: empirical transport vs.\ integrated connection}
\label{sec:attribution}

A sharper test of the geometric interpretation bypasses the transport
factors entirely.
For each detected loop, \texttt{integrate\_transport()} numerically
integrates the analytic connection \eqref{eq:transport} along the actual
chain trajectory between the loop endpoints, producing a per-loop,
per-group \emph{prediction} of the fiber displacement attributable to
parallel transport, together with the signed enclosed area of the base
trajectory.
Holonomy predicts displacements whose sign tracks loop orientation and
whose magnitude tracks the enclosed area and local curvature $F_j$;
generic autocorrelation predicts displacements uncorrelated with
orientation and area.

On the sparse GLMM benchmark analysed here, this attribution test yields a
null result.
Regressing empirical displacement on the integrated-connection prediction
gives slopes near zero for all groups (range $-0.001$ to $0.005$; one
borderline significant at $p = 0.025$), and empirical displacements show no
dependence on signed loop area in any group.

Two caveats scope this null.
\textit{Test harshness}: integrated parallel-transport predictions reach
$|\Delta\alpha| \approx 100$--$200$ against observed displacements of
$2$--$5$, because the linear-transport approximation breaks down along
long, jagged MCMC paths through high-curvature regions.
A restricted test on only the shortest, tightest loops (length $\leq 50$,
tightest distance quartile) also yields null results: 38 loops survived
the filter and per-group slopes are of order $10^{-5}$ per posterior SD.

\textit{Negative control}: for each gap $g$, we construct matched-gap
control pairs chosen to have \emph{large} base distance (above the
median of pairwise distances at that gap).
On the prior-dominated cell ($n_j = 3$, $\sigma_\text{true} = 0.5$,
high $\pi_j$), true loops give $\bar{\rho}_j \approx 0.03$--$0.07$
while controls give $\bar{\rho}_j \approx 0$: base-closure conditioning
carries signal in exactly the cell where $\pi_j$ predicts it.
On the likelihood-dominated cell ($\sigma_\text{true} = 3.0$, low $\pi_j$),
both arms give $\bar{\rho}_j \approx 0$.
Figure~\ref{fig:control_pairs} shows the full gap profile for all four
experimental cells: the contrast between true-loop $\hat{\rho}_j$ and the
near-zero control band is confined to the high-$\pi_j$ cell, in quantitative
agreement with the analytic prediction.

This null is exactly what the flatness of the connection
(Proposition~\ref{prop:flat}) requires: there is no curvature, so there is no
parallel-transport displacement for the chain to track, whether or not the
sampler follows horizontal lifts.
The geometry's chain-level signature is instead \emph{excess conditional
autocorrelation localised in prior-dominated groups}---precisely what
base-closure conditioning isolates and what the analytic $\pi_j$ predicts.
The analytic $\pi_j$ is the primary and more robust diagnostic output: it
requires no chain-based estimation, is free of residualisation leakage, and
directly identifies which groups will show elevated $\hat{\rho}_j$.
The chain-based $\hat{\rho}_j$ and its gap-sensitivity profile serve as
empirical corroboration of the $\pi_j$ prediction, confirming that the
conditional dependence is detectable in actual chains, but $\pi_j$ remains
the headline flag.
The diagnostic workflow (loop detection, estimation of $\hat{\rho}_j$, and
gap-sensitivity profiling) remains a valid probe of loop-conditional fiber
dependence; flatness tells us that what it measures is conditional
autocorrelation rather than holonomy.

\begin{figure}[ht]
  \centering
  \includegraphics[width=0.92\textwidth]{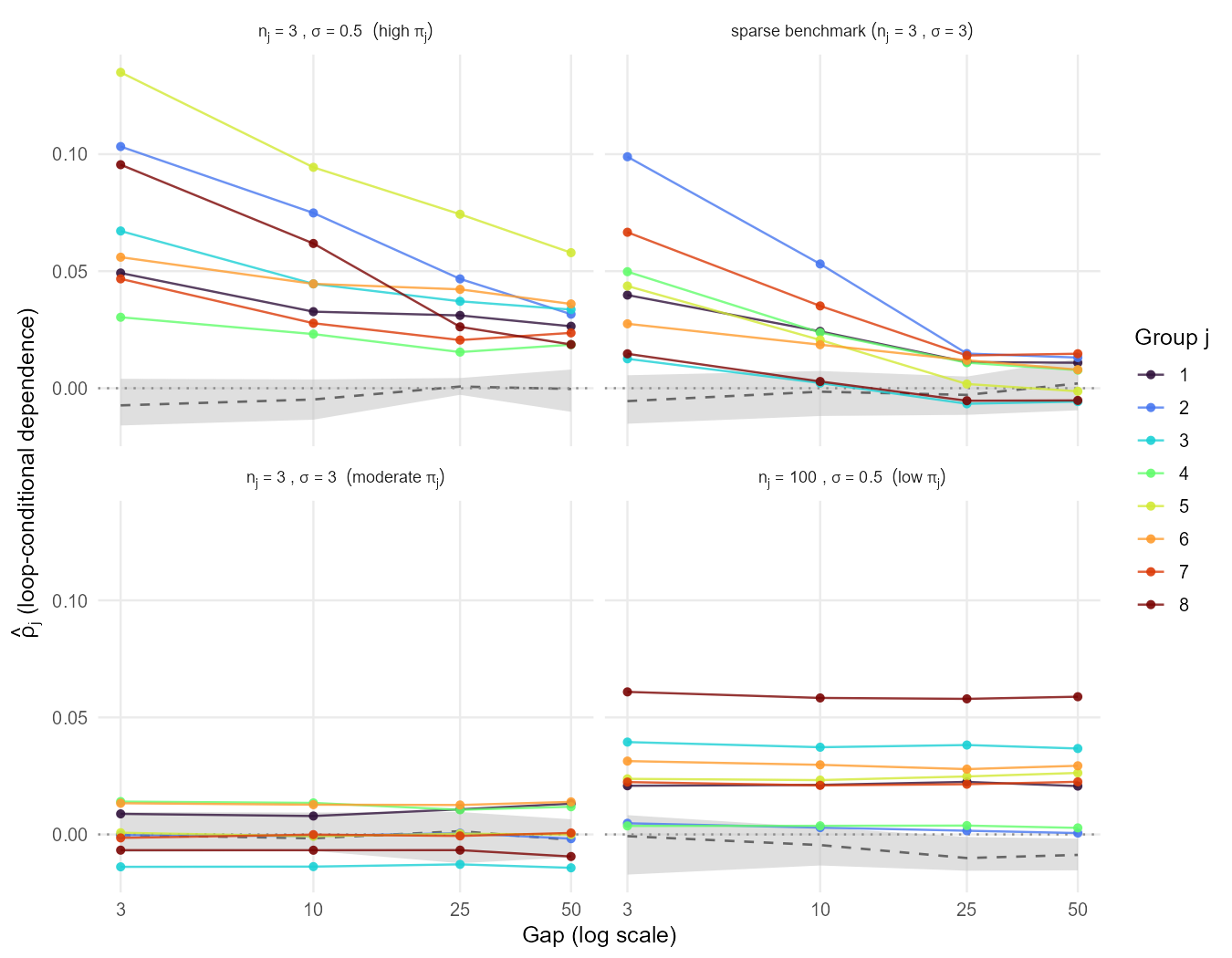}
  \caption{Negative-control experiment: true-loop $\hat{\rho}_j$ vs.\
    matched-gap control pairs, per group (colour) across the gap grid.
    \emph{Coloured lines and points}: per-group $\hat{\rho}_j$ from true
    loops (small base distance, base-closure conditioning active).
    \emph{Grey band}: full range of $\hat{\rho}_j$ from matched-gap
    control pairs with \emph{large} base distance; dashed line is the
    control mean.
    Panels are ordered from high $\pi_j$ (top left) to low $\pi_j$
    (bottom right).
    In the high-$\pi_j$ cell ($n_j=3$, $\sigma=0.5$) the true-loop
    $\hat{\rho}_j$ values are consistently above the control band across
    all gaps, confirming that base-closure conditioning carries signal
    beyond lag autocorrelation and residualisation leakage in exactly
    the cell where $\pi_j$ predicts it.
    In the moderate- and low-$\pi_j$ cells both arms are indistinguishable
    from zero.}
  \label{fig:control_pairs}
\end{figure}

\subsection{Comparison: centred vs.\ non-centred}

Figure~\ref{fig:comparison} shows the per-group dependence coefficients across
the gap grid for the centred and non-centred parameterisations of a sparse
logistic GLMM ($J = 8$ groups, $n_j = 3$ observations each,
$\sigma_\text{true} = 3$).
At the shortest gap ($g = 3$), the centred chain shows a modestly larger
mean coefficient ($\bar{\rho} \approx 0.044$) than the non-centred
chain ($\bar{\rho} \approx 0.022$).
The centred profile decays with gap (to $\approx 0.005$ at $g = 50$),
while the non-centred profile is flat across all gaps
($\approx 0.022$--$0.023$).
The gap-independence of the non-centred factors, for a parameterisation
in which $\tilde{\alpha}_j$ is conditionally independent of $(\mu,\sigma)$
by construction, is consistent with residualisation leakage rather than
geometric signal (Section~\ref{sec:attribution}).
At long gaps the centred profile falls below the non-centred level;
the short-gap excess in the centred chain ($\approx 0.044$ at $g=3$ vs
$\approx 0.005$ at $g=50$) is the autocorrelation contribution that decays
with lag, while the flat non-centred baseline ($\approx 0.022$--$0.023$
across all gaps) is the leakage floor, confirmed by the negative-control
experiment.

\begin{figure}[ht]
  \centering
  \includegraphics[width=0.95\textwidth]{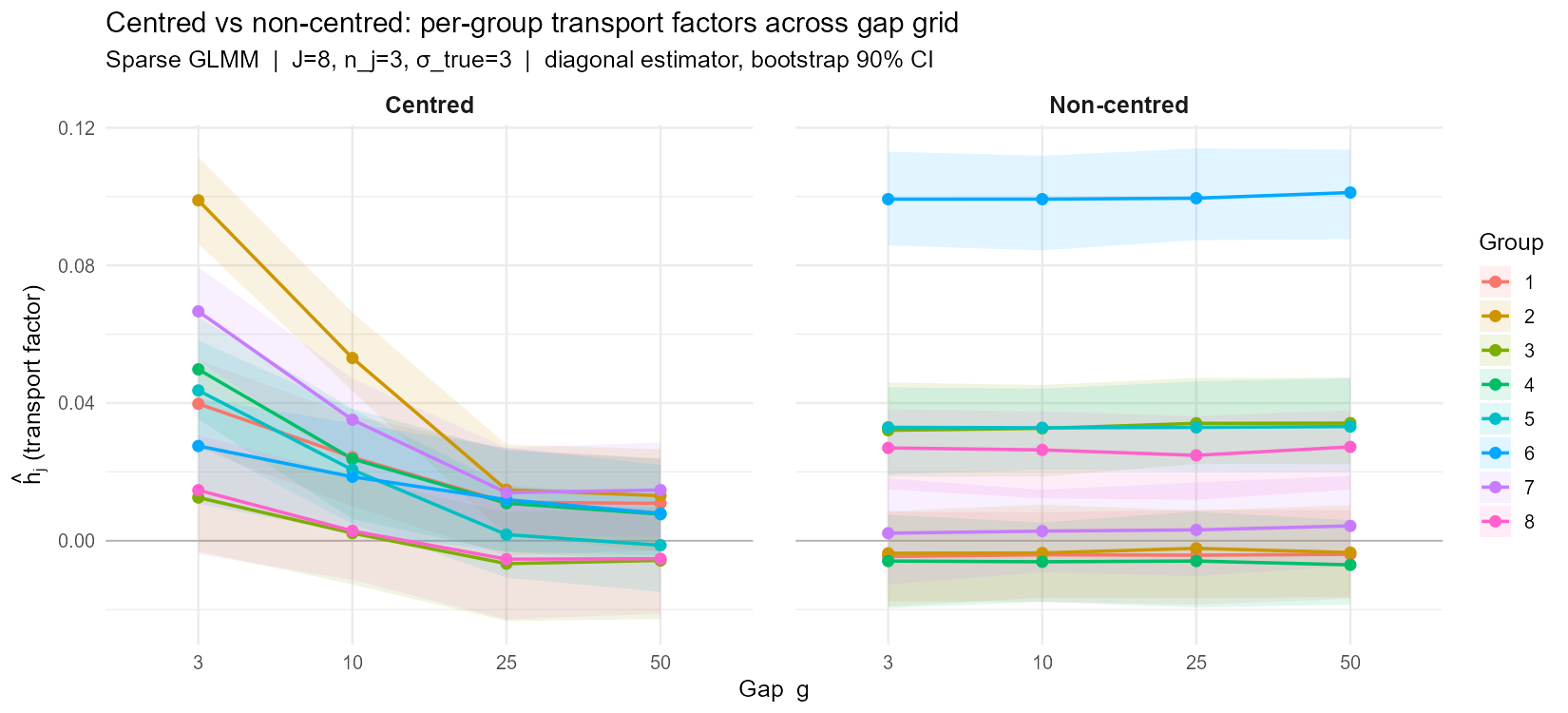}
  \caption{Per-group dependence coefficients $\hat{\rho}_j$ with bootstrap 90\%
    intervals for the centred (left) and non-centred (right)
    parameterisations of the sparse logistic GLMM ($J=8$, $n_j=3$,
    $\sigma_\text{true}=3$), shown across the gap grid
    $g \in \{3, 10, 25, 50\}$.
    The centred chain has additional short-gap signal that decays with lag
    ($\bar{\rho} \approx 0.044$ at $g=3$ falling to $\approx 0.005$ at $g=50$).
    The non-centred chain shows a flat, gap-stable baseline
    ($\bar{\rho} \approx 0.022$--$0.023$ at all gaps), consistent with
    residualisation leakage rather than geometric signal
    (Section~\ref{sec:attribution}).}
  \label{fig:comparison}
\end{figure}

\section{Analytic Connection for the Centred Logistic GLMM}
\label{sec:glmm_connection}

\subsection{Model}

We work throughout with the two-level logistic GLMM:
\begin{align}
  y_{ji} &\sim \Bern\!\left(\text{logit}^{-1}(\alpha_j + \bm{x}_{ji}^\top\bm{\beta})\right),
    \quad j \in [J],\; i = 1,\ldots,n_j, \label{eq:glmm_lik} \\
  \alpha_j &\sim \Normal(\mu, \sigma^2), \label{eq:glmm_prior} \\
  \mu &\sim \Normal(0, 25), \quad \sigma \sim \text{Exponential}(1), \quad
  \bm{\beta} \sim \Normal(\bm{0}, 4I). \label{eq:glmm_hyperprior}
\end{align}
The base space is $\Bs = \R \times \R_{>0}$ with coordinates $(\mu, \sigma)$,
and the fiber is $\R^J$ with coordinates $\bm{\alpha} = (\alpha_1,\ldots,\alpha_J)$.
We write $p_{ji} = \text{logit}^{-1}(\alpha_j + \bm{x}_{ji}^\top\bm{\beta})$
for the success probability of observation $i$ in group $j$.

\subsection{Fisher metric blocks}

We instantiate the block metric \eqref{eq:G_block} for the GLMM.
Because the logistic link is canonical, the observed information used below
equals the expected Fisher information (Section~\ref{sec:geometry}), so naming
these ``Fisher metric blocks'' is unambiguous here.

\begin{proposition}[Fisher metric blocks for the centred GLMM]
\label{prop:metric_blocks}
  The negative Hessian of the log-posterior for model
  \eqref{eq:glmm_lik}--\eqref{eq:glmm_hyperprior} has the following blocks,
  evaluated at a single parameter point $(\mu, \sigma, \bm{\alpha}, \bm{\beta})$:
  \begin{align}
    G_{FF,jj} &= \frac{1}{\sigma^2} + \sum_{i=1}^{n_j} p_{ji}(1-p_{ji}),
    \label{eq:GFF} \\
    G_{BF,j1} &= -\frac{1}{\sigma^2},
    \quad\text{(base direction: $\mu$)} \label{eq:GBF_mu} \\
    G_{BF,j2} &= -\frac{2(\alpha_j - \mu)}{\sigma^3},
    \quad\text{(base direction: $\sigma$)} \label{eq:GBF_sigma}
  \end{align}
  where $G_{FF}$ is diagonal and $G_{BF} \in \R^{J \times 2}$.
\end{proposition}

\begin{proof}
  The log-posterior is
  \[
    \log p \propto
    \sum_{j \in [J]}\sum_{i=1}^{n_j} \bigl[y_{ji}\log p_{ji} + (1-y_{ji})\log(1-p_{ji})\bigr]
    - \sum_j \frac{(\alpha_j - \mu)^2}{2\sigma^2}
    - J\log\sigma - \text{priors on } \mu, \sigma, \bm{\beta}.
  \]
  Taking second derivatives: the cross term between $\alpha_j$ and $\alpha_k$
  ($j \neq k$) vanishes because observations within group $j$ do not depend on
  $\alpha_k$.
  The diagonal $\partial^2/\partial\alpha_j^2$ gives \eqref{eq:GFF}.
  Cross terms $\partial^2/\partial\alpha_j\,\partial\mu$ and
  $\partial^2/\partial\alpha_j\,\partial\sigma$ give \eqref{eq:GBF_mu}
  and \eqref{eq:GBF_sigma} respectively,
  from the prior term $-(\alpha_j - \mu)^2/(2\sigma^2)$.
\end{proof}

\begin{remark}
  The full base--base block $G_{BB}$ is not positive definite outside the
  posterior mode in the centred parameterisation.
  This is the fingerprint of the funnel: the realised geometry of the total
  space is indefinite in the $(\mu, \sigma)$ directions when the chain is in the
  neck of the funnel.
  The indefiniteness is confined to $G_{BB}$, however; the fiber block
  $G_{FF,jj} = 1/\sigma^2 + \sum_i p_{ji}(1-p_{ji})$ is strictly positive
  throughout, so the connection $\Aconn = -\GFF^{-1}\GBF$ and its
  $G$-orthogonality are globally well-defined and never see the funnel.
  This is the two-pathologies distinction already visible in the metric: the
  fiber geometry is always benign, the funnel is purely base-space.
  The SoftAbs regularisation \citep{betancourt2013} resolves the base
  indefiniteness for the RMHMC sampler (Section~\ref{sec:rmhmc_bg}).
\end{remark}

\subsection{Connection form}

From \eqref{eq:connection} and Proposition~\ref{prop:metric_blocks},
the connection coefficients are:
\begin{align}
  A_{j,\mu} &= \frac{1}{\sigma^2 G_{FF,j}},
  \label{eq:A_mu} \\
  A_{j,\sigma} &= \frac{2(\alpha_j - \mu)}{\sigma^3 G_{FF,j}}.
  \label{eq:A_sigma}
\end{align}
Both depend on the current fiber point $\alpha_j$, confirming that the
connection is genuinely non-linear.

\subsection{Curvature: linearised and full}

\begin{proposition}[Curvature of the linearised connection]
\label{prop:curvature}
  Let $\tilde\Aconn(\bm\theta) := \Aconn(\bm\theta, \bm\alpha_0)$ be the
  fiber-frozen linearisation of $\Aconn$ at a fixed fiber point
  $\bm\alpha_0$.
  Its curvature \eqref{eq:curvature}, for the centred GLMM, is
  \begin{equation}
    F_j = \frac{\partial A_{j,\sigma}}{\partial \mu}
          - \frac{\partial A_{j,\mu}}{\partial \sigma}
        = -\frac{2}{\sigma^5 G_{FF,j}^2}.
    \label{eq:curvature_glmm}
  \end{equation}
  Equivalently, $F_j$ is the exterior derivative of $\Aconn$ taken with the
  fiber held fixed.
\end{proposition}

\begin{proof}
  From \eqref{eq:A_sigma}, noting that $G_{FF,j}$ does not depend on $\mu$
  (the likelihood term involves $\alpha_j$, $\bm{\beta}$, and the data, but not $\mu$):
  \[
    \frac{\partial A_{j,\sigma}}{\partial \mu}
    = \frac{-2}{\sigma^3 G_{FF,j}}.
  \]
  From \eqref{eq:A_mu}, noting that $\partial G_{FF,j}/\partial\sigma = -2/\sigma^3$
  (from the prior term $1/\sigma^2$):
  \[
    \frac{\partial A_{j,\mu}}{\partial \sigma}
    = \frac{-2}{\sigma^3 G_{FF,j}}
      + \frac{1}{\sigma^2} \cdot \frac{2/\sigma^3}{G_{FF,j}^2}
    = \frac{-2}{\sigma^3 G_{FF,j}} + \frac{2}{\sigma^5 G_{FF,j}^2}.
  \]
  Subtracting: $F_j = -2/(\sigma^5 G_{FF,j}^2)$.
\end{proof}

This linearised curvature is always negative ($F_j < 0$), and $|F_j|$ grows
as the prior dominates (small $\sigma$, few data).
Read naively, this suggests a contractive holonomy of $\tilde\Aconn$ that
worsens as the prior tightens.
But $\tilde\Aconn$ is not the connection the sampler would have to follow; the
true connection is the fiber-dependent $\Aconn$.
By Proposition~\ref{prop:universal_flat} its full curvature vanishes; for the
GLMM this can be seen explicitly, and the cancellation is instructive.

\begin{proposition}[Flatness of the GLMM connection, instance of
Proposition~\ref{prop:universal_flat}]
\label{prop:flat}
  The full curvature of the metric-orthogonal Ehresmann connection $\Aconn$ of
  \eqref{eq:A_mu}--\eqref{eq:A_sigma}, including the vertical
  (fiber-derivative) terms, is identically zero for every group $j$:
  \begin{equation}
    F_j^{\mathrm{full}}
    = \underbrace{\frac{\partial A_{j,\sigma}}{\partial \mu}
      - \frac{\partial A_{j,\mu}}{\partial \sigma}}_{=\,-2/(\sigma^5 G_{FF,j}^2)}
    + \underbrace{A_{j,\mu}\frac{\partial A_{j,\sigma}}{\partial \alpha_j}
      - A_{j,\sigma}\frac{\partial A_{j,\mu}}{\partial \alpha_j}}_{=\,+2/(\sigma^5 G_{FF,j}^2)}
    = 0 .
    \label{eq:flat}
  \end{equation}
  Consequently the horizontal distribution is integrable (Frobenius), and
  since the base $\Bs = \R \times \R_{>0}$ is simply connected, the holonomy of
  every loop is trivial: a global flat trivialisation exists.
\end{proposition}

\begin{proof}
  The first pair is Proposition~\ref{prop:curvature}.
  For the vertical pair, write $S_j = \sum_i p_{ji}(1-p_{ji})$ so that
  $G_{FF,j} = 1/\sigma^2 + S_j$ and
  $\partial G_{FF,j}/\partial\alpha_j = S_j' = \sum_i p_{ji}(1-p_{ji})(1-2p_{ji})$.
  Then
  \[
    A_{j,\mu}\frac{\partial A_{j,\sigma}}{\partial \alpha_j}
    = \frac{1}{\sigma^2 G_{FF,j}}\!\left[\frac{2}{\sigma^3 G_{FF,j}}
      - \frac{2(\alpha_j-\mu)S_j'}{\sigma^3 G_{FF,j}^2}\right]
    = \frac{2}{\sigma^5 G_{FF,j}^2}
      - \frac{2(\alpha_j-\mu)S_j'}{\sigma^5 G_{FF,j}^3},
  \]
  \[
    A_{j,\sigma}\frac{\partial A_{j,\mu}}{\partial \alpha_j}
    = \frac{2(\alpha_j-\mu)}{\sigma^3 G_{FF,j}}
      \!\left[-\frac{S_j'}{\sigma^2 G_{FF,j}^2}\right]
    = -\frac{2(\alpha_j-\mu)S_j'}{\sigma^5 G_{FF,j}^3}.
  \]
  The $S_j'$ terms cancel in the difference, leaving
  $+2/(\sigma^5 G_{FF,j}^2)$, which is exactly the negative of the first pair.
  Hence $F_j^{\mathrm{full}} = 0$.
\end{proof}

\begin{remark}[Why the linearised curvature is nonzero]
\label{rem:leading_order}
  The discrepancy between $F_j \neq 0$ and $F_j^{\mathrm{full}} = 0$ is the
  vertical variation of $\Aconn$ along the transported fiber.
  Freezing the fiber (equivalently, holding $G_{FF,j}$ at its value at the loop
  centre) discards the second pair in \eqref{eq:flat} and recovers the nonzero
  $F_j$.
  This is precisely what the synthetic-loop validation of
  Section~\ref{sec:stokes_validation} does, so the ``holonomy'' it reports is
  the holonomy of the linearised connection, not of the true one.
  We confirm the result numerically in
  \texttt{data-raw/verify\_flat\_connection.R}: integrating the true transport
  ODE around closed loops returns displacement of order $10^{-14}$, while
  freezing $G_{FF,j}$ reproduces the linearised holonomy of
  Figure~\ref{fig:stokes}.
  Both $F_j$ and $F_j^{\mathrm{full}}$ are curvatures of the Ehresmann
  connection (iii) of Section~\ref{sec:which_connection}; neither is the
  Riemann curvature of the Fisher--Rao metric, which is hyperbolic and nonzero
  and which the flatness of $\Aconn$ does not affect.
\end{remark}

\begin{remark}[The structure group is abelian: no rotational holonomy]
\label{rem:abelian}
  Because $G_{FF}$ is diagonal and the rows of $G_{BF}$ decouple across
  groups, parallel transport acts on each fiber coordinate independently: any
  loop's transport is a per-group scaling $\alpha_j \mapsto h_j \alpha_j$ with
  $h_j \in \R_{>0}$ (and, by Proposition~\ref{prop:flat}, $h_j = 1$ for the
  true connection).
  The structure group of this bundle is abelian, and genuinely
  \emph{rotational} holonomy (complex transport eigenvalues) is impossible
  for this model class.
  This has a practical consequence for estimation:
  the diagnostic's dependence matrix $\hat{P}$ is diagonal with real entries,
  so it should estimate $J$ scalars rather than a full $J \times J$ matrix
  (Section~\ref{sec:transport_estimation}).
\end{remark}

\subsection{Prior fraction}
\label{sec:prior_fraction}

The ratio of the prior precision to the total fiber precision
\begin{equation}
  \pi_j = \frac{1/\sigma^2}{G_{FF,j}}
         = \frac{1/\sigma^2}{1/\sigma^2 + \sum_{i=1}^{n_j}p_{ji}(1-p_{ji})}
  \label{eq:prior_fraction}
\end{equation}
takes values in $(0,1)$ and measures how much of the fiber metric at group $j$
is attributable to the prior.
When $\pi_j \to 1$ the prior dominates (sparse data or large $\sigma$);
when $\pi_j \to 0$ the likelihood dominates.
This is exactly the \emph{pooling factor} $\omega$ of \citet{gelman2006pooling},
the fraction of a group's posterior precision contributed by the prior: their
$\omega = 1 - \sigma_\alpha^2/(\sigma_\alpha^2 + \sigma_y^2)$ is
$(1/\sigma_\alpha^2)/(1/\sigma_\alpha^2 + n_j/\sigma_y^2)$, which is
Equation~\eqref{eq:prior_fraction} for the Gaussian case. It is thus a classical quantity, not a new one:
recoverable by hand from a fitted multilevel model's variance components and
per-group information, and reported in related (population-level) form as the
intraclass correlation by standard tools. What we add is its geometric derivation
as the residual obstruction once the connection is flat, and its per-coordinate
computation; the package simply returns it. (We use ``pooling
factor'' rather than ``shrinkage factor'' deliberately: some authors call the
complement $1-\pi_j$ the shrinkage factor, a usage \citeauthor{gelman2006pooling}
flag as confusing, since a shrinkage factor of zero then means complete pooling.)
It is also the per-group form of the prior/likelihood balance that
\citet{betancourt2015} tied to the optimal parameterisation.
In plain terms, $\pi_j$ is the share of what the posterior knows about group $j$
that comes from the prior rather than from that group's own data: near 1, the
data barely constrain the group and a non-centred parameterisation helps; near
0, the group is well identified and centring is fine.
What the bundle derivation adds locally is the observation that $\pi_j$ is
determined by prior \emph{precision} $1/\sigma^2$ rather than prior
\emph{variance} --- the source of the non-monotone behaviour in $\sigma$
documented in Section~\ref{sec:simulation}.
Substituting into the linearised curvature \eqref{eq:curvature_glmm} gives
$|F_j| = 2\pi_j^2/\sigma$: even the spurious linearised curvature is controlled
by $\pi_j$, so the diagnostic and the analytic flag agree on where any signal
must concentrate.

\begin{remark}
  The prior fraction $\pi_j$ is computed analytically from a subsample of
  posterior draws and serves as a diagnostic in its own right:
  groups with $\pi_j > 0.6$ are good candidates for non-centred
  reparameterisation (Section~\ref{sec:reparam}).
  This is the geometric basis for the per-group parameterisation advice
  implemented in \texttt{fibr}.
\end{remark}

\subsection{Synthetic loop validation}
\label{sec:stokes_validation}

We check the \emph{linearised} curvature formula \eqref{eq:curvature_glmm}
against Stokes' theorem \eqref{eq:stokes} using synthetic circular loops in
$(\mu, \sigma)$ space, bypassing any MCMC noise.
We stress that this validates the linearised connection only: the integration
below holds the fiber fixed at $\bm{\alpha}_0$ (equivalently, freezes
$G_{FF}$), which is the linearisation whose curvature is nonzero.
Integrating the \emph{true} connection, with $G_{FF}$ updated as the fiber
moves, returns zero displacement to integrator precision, consistent with
Proposition~\ref{prop:flat}.

For a circle of radius $r$ centred at the posterior mean
$(\mu_0, \sigma_0, \bm{\alpha}_0, \bm{\beta}_0)$:
\begin{enumerate}
  \item Numerically integrate the ODE \eqref{eq:transport} around the circle
    (discretised into $n_s = 400$ steps) to obtain $\Delta\alpha_j^{\text{num}}$.
  \item Compute the first-order Stokes prediction
    $\Delta\alpha_j^{\text{Stokes}} = F_j(\mu_0,\sigma_0,\bm{\alpha}_0) \cdot \pi r^2$.
\end{enumerate}

Figure~\ref{fig:stokes} shows the comparison for radii $r \in [0.02, 0.30]$
(approximately $0.02$--$0.30$ posterior standard deviations), displayed
separately for each of the $J = 8$ groups.
The numerically integrated linearised transport and its first-order Stokes
prediction agree closely for $r \leq 0.20$ across all groups, and diverge for
larger radii where the quadratic Stokes term becomes relevant.
The divergence is not uniform: groups with higher $|F_j|$ depart from the
linear Stokes prediction at smaller $r$, while groups with near-flat linearised
curvature remain in agreement throughout the range.
This agreement validates the internal consistency of the connection
\eqref{eq:A_mu}--\eqref{eq:A_sigma} and the linearised curvature
\eqref{eq:curvature_glmm}.
It is \emph{not} evidence of geometric holonomy: by Proposition~\ref{prop:flat}
the true holonomy is zero, and Figure~\ref{fig:stokes} reports the linearised
object that the freezing of $\bm{\alpha}_0$ creates.

\begin{figure}[ht]
  \centering
  \includegraphics[width=0.75\textwidth]{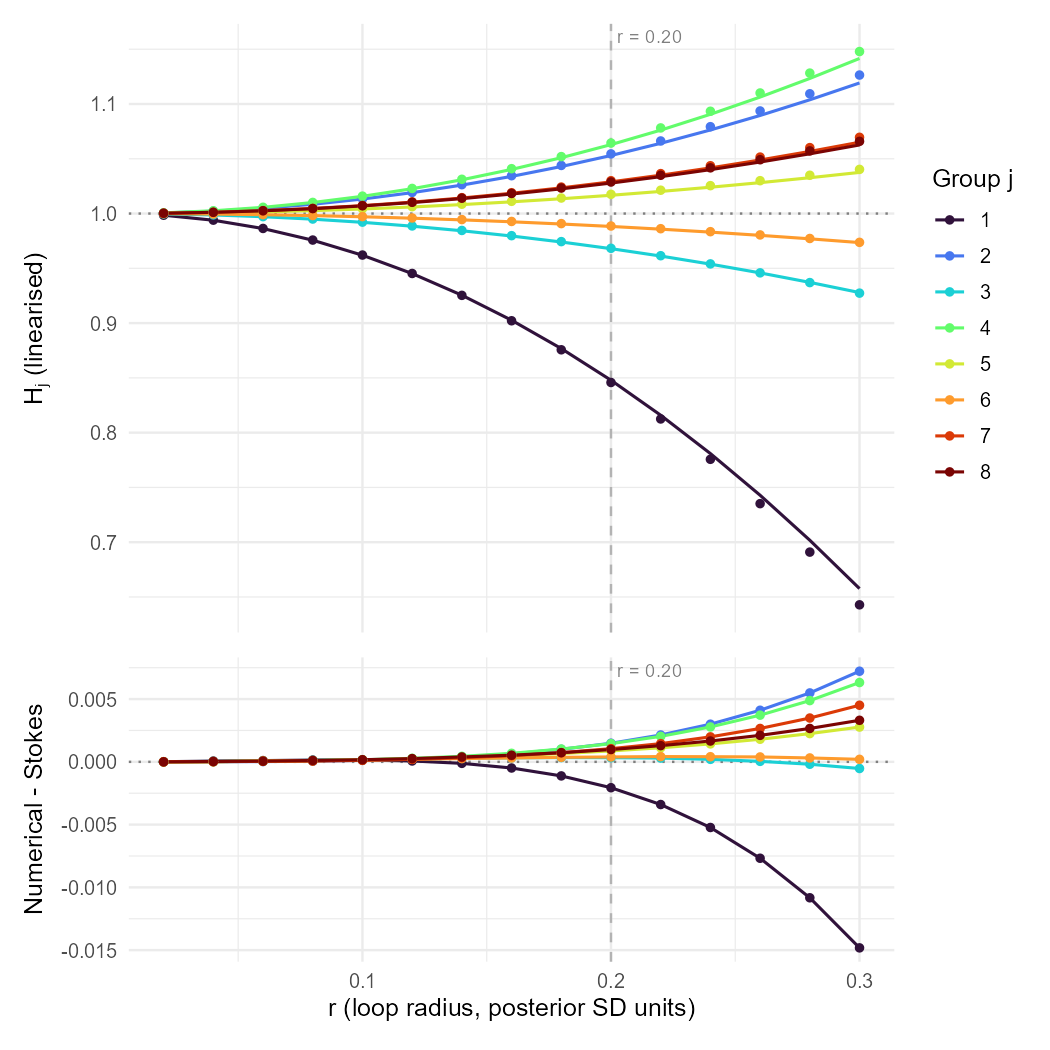}
  \caption{\emph{Linearised} holonomy $H_j$ vs.\ loop radius $r$ for the
    centred GLMM, shown per group ($j = 1, \ldots, 8$, colour).
    \emph{Top}: lines are the first-order Stokes prediction
    $H_j = 1 + F_j \cdot \pi r^2 / \alpha_{0j}$; points are numerical
    integration of the transport ODE with the fiber held fixed at
    $\bm{\alpha}_0$ (i.e.\ $G_{FF}$ frozen).
    Both objects are the \emph{linearised} holonomy: with $G_{FF}$ updated
    along the fiber the displacement is zero (Proposition~\ref{prop:flat}).
    \emph{Bottom}: residuals (numerical $-$ Stokes) reveal the structure of
    the agreement.  Residuals are near zero for $r \leq 0.20$ across all
    groups, confirming internal consistency of the linearised curvature
    formula; the negative departure at larger $r$ is the quadratic Stokes
    correction term, sharpest for groups with high $|F_j|$.
    The vertical dashed line marks $r = 0.20$.}
  \label{fig:stokes}
\end{figure}

\section{Simulation Study}
\label{sec:simulation}

\subsection{Design}

We ran a $4 \times 4$ factorial simulation study varying observations per group
$n_j \in \{3, 10, 30, 100\}$ and true hierarchical standard deviation
$\sigma_\text{true} \in \{0.5, 1.0, 2.0, 3.0\}$, with $J = 8$ groups and
10 replicates per cell (160 Stan fits in total).
For each replicate: data were simulated from model
\eqref{eq:glmm_lik}--\eqref{eq:glmm_hyperprior} with
$\bm{\beta}_\text{true} = (0.8, -0.5)$, $\mu_\text{true} = 0$;
the centred Stan model was fitted with 4 chains $\times$ 2{,}000 post-warmup
iterations; the dependence diagnostic (diagonal estimator) was evaluated at
each gap in the grid $g \in \{3, 10, 25, 50\}$, alongside the estimated
IACT of the base and fiber chains;
and the analytic prior fraction was computed from 200 subsampled draws via
\texttt{compute\_connection()}.

The primary scalar is the mean dependence coefficient
$\bar{\rho} = J^{-1}\sum_j \hat{\rho}_j$ at each gap.
As established in Section~\ref{sec:diagnostic}, this is near 0 when loop
endpoints are independent (well-mixed at the loop time-scale) and grows
toward 1 with loop-conditional persistence.
The gap profile separates the two contributors that flatness leaves: ordinary
autocorrelation, which decays once $g$ exceeds the IACT, and residualisation
leakage, which is gap-independent (Section~\ref{sec:attribution}).
With flatness established analytically, the simulation is not a search for a
geometric signal but a check of two things: that the conditional dependence
$\pi_j$ predicts is detectable and correctly localised in finite chains, and
that it is geometrically separable from the funnel.

\subsection{Results}

The headline of the study is the geometric separation of the two pathologies
(the non-monotone $\sigma$ behaviour below); the theory--empirical agreement
confirms that the predicted dependence is detectable in finite chains, but it
is the weaker, expected half of the picture.

\paragraph{Heatmap.}
Figure~\ref{fig:heatmap} shows the median $\bar{\rho}$ over 10 replicates.
The signal is strongest in the upper-left corner ($n_j = 3$,
$\sigma_\text{true} = 0.5$--$1.0$) and decreases monotonically with $n_j$
at every $\sigma_\text{true}$ value.

\paragraph{The $\sigma$-direction is non-monotone and illuminating.}
Within the $n_j = 3$ row, $\bar{\rho}$ is highest at
$\sigma_\text{true} = 0.5$--$1.0$ (0.040 and 0.036 at the headline gap $g = 50$)
and drops sharply for $\sigma_\text{true} \in \{2.0, 3.0\}$ (0.013, 0.015).
This is not a failure of the diagnostic but a genuine distinction:
the prior fraction $\pi_j = (1/\sigma^2)/G_{FF,j}$ is determined by the ratio
of prior \emph{precision} to total precision.
Increasing $\sigma$ reduces $1/\sigma^2$, which \emph{reduces} the prior fraction
(Table~\ref{tab:summary}) even as it worsens the funnel.
The funnel pathology (large $\sigma$, slow mixing in $\sigma$) and the
coupling pathology (large $\pi_j$, strong base-conditional fiber dependence)
are related but distinct: the former is a problem of base-space geometry; the
latter is a problem of the base--fiber coupling.
Standard diagnostics (divergences, low ESS) detect the former;
the per-group $\hat{\rho}_j$ and analytic $\pi_j$ from \texttt{fibr} localise
the coupling pathology to specific groups, its chain-level signature being
structured conditional autocorrelation, per Section~\ref{sec:attribution}.

\begin{figure}[ht]
  \centering
  \includegraphics[width=0.80\textwidth]{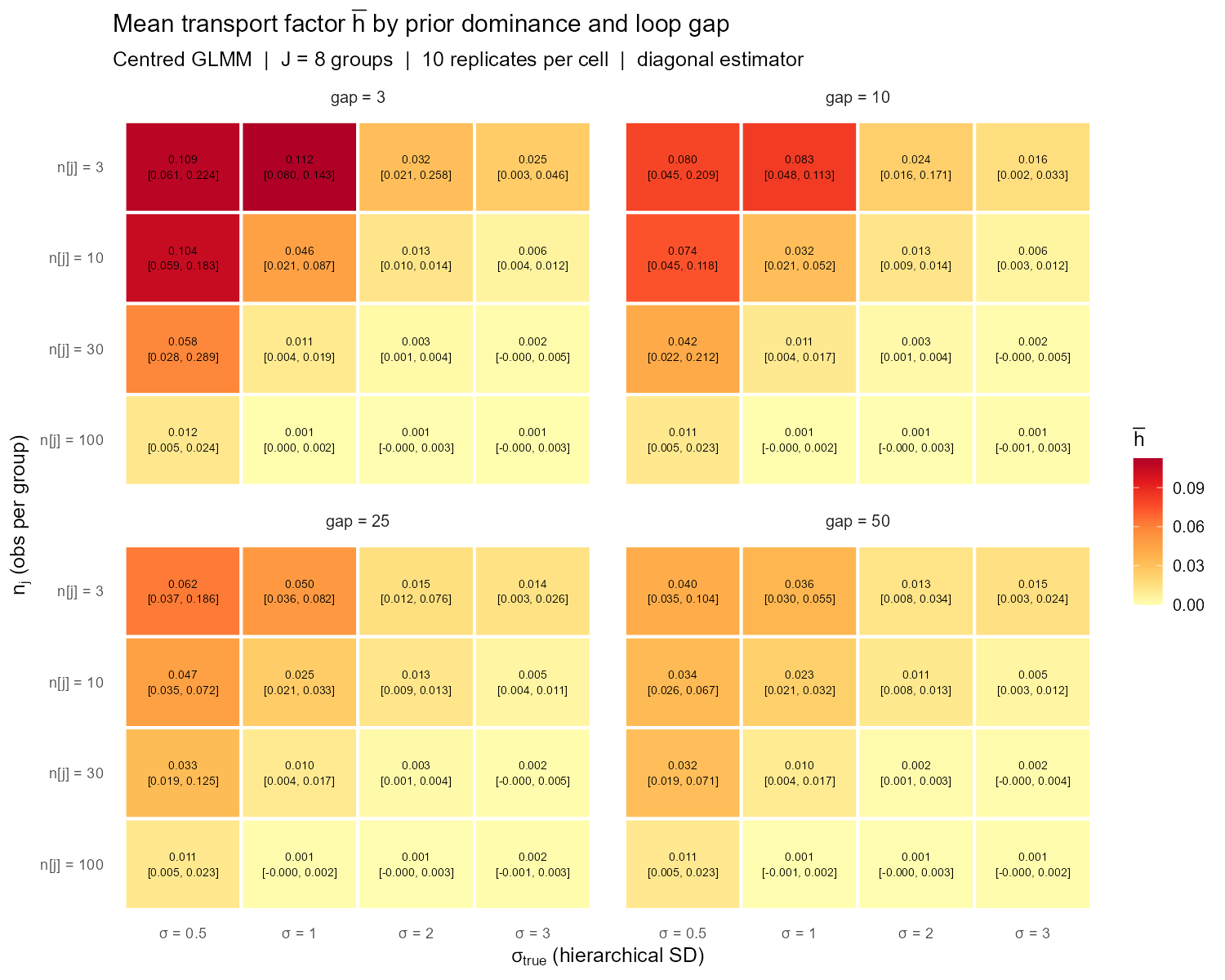}
  \caption{Simulation study: median mean dependence coefficient
    $\bar{\rho}$ over 10 replicates, by gap $g$.
    Tile labels show median [Q25, Q75].
    The signal (dark tiles) is concentrated in the prior-dominated
    upper-left corner ($n_j$ small, $\sigma_\text{true}$ small to moderate).
    The $\sigma$-axis is non-monotone: the coupling signal is driven by prior
    \emph{precision} $1/\sigma^2$, not by $\sigma$ itself.}
  \label{fig:heatmap}
\end{figure}

\paragraph{Theory vs.\ empirical.}
Figure~\ref{fig:theory_vs_empirical} plots the mean prior fraction $\bar{\pi}$
(analytic) against $\bar{\rho}$ (diagnostic), faceted by gap, for all 640
cell-replicate--gap combinations.
The LOESS trend is positive, consistent with the analytic prior fraction
being a leading predictor of the empirical signal.
Given flatness, we read this straightforwardly: prior dominance slows mixing,
so $\pi_j$ predicts the strength of the conditional autocorrelation that the
diagnostic measures, with no geometric component to disentangle.
The attribution test of Section~\ref{sec:attribution} yielded the
corresponding null result; this figure establishes that the analytic quantity
and the empirical diagnostic move together across the design.
Scatter around the trend reflects simulation variability in the data and chain.

\begin{figure}[ht]
  \centering
  \includegraphics[width=0.80\textwidth]{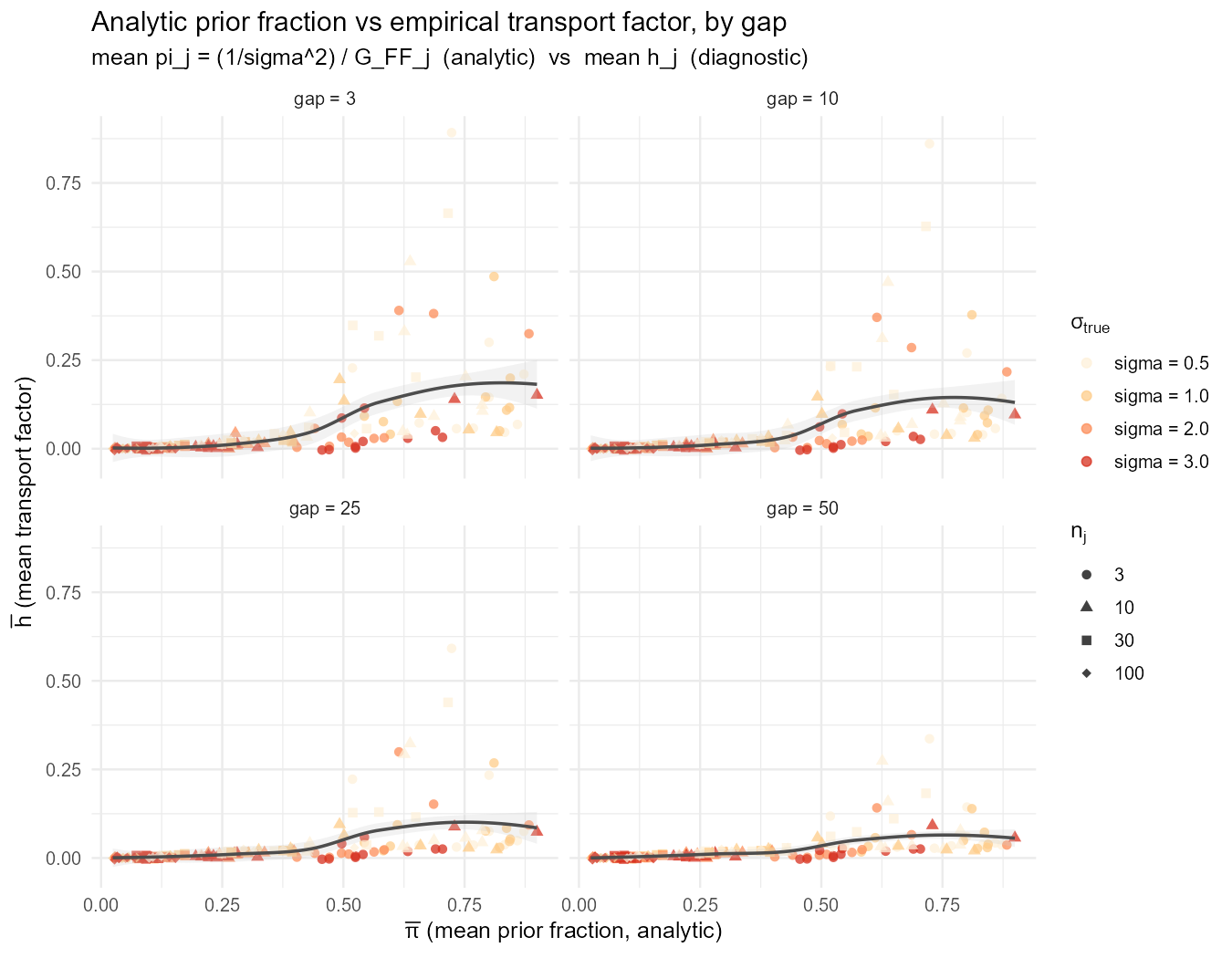}
  \caption{Analytic prior fraction $\bar{\pi}$ (x-axis) against empirical
    mean dependence coefficient $\bar{\rho}$ (y-axis), faceted by loop gap $g$.
    Each panel contains all 160 cell-replicates at that gap.
    Colour encodes $\sigma_\text{true}$; shape encodes $n_j$.
    The positive LOESS trend (grey band: 95\% pointwise CI) is consistent
    with the prediction that a larger prior fraction implies stronger
    conditional dependence.  The trend persists at all gaps but attenuates as
    $g$ increases, reflecting the decay of the autocorrelation contribution.}
  \label{fig:theory_vs_empirical}
\end{figure}

\begin{figure}[ht]
  \centering
  \includegraphics[width=0.80\textwidth]{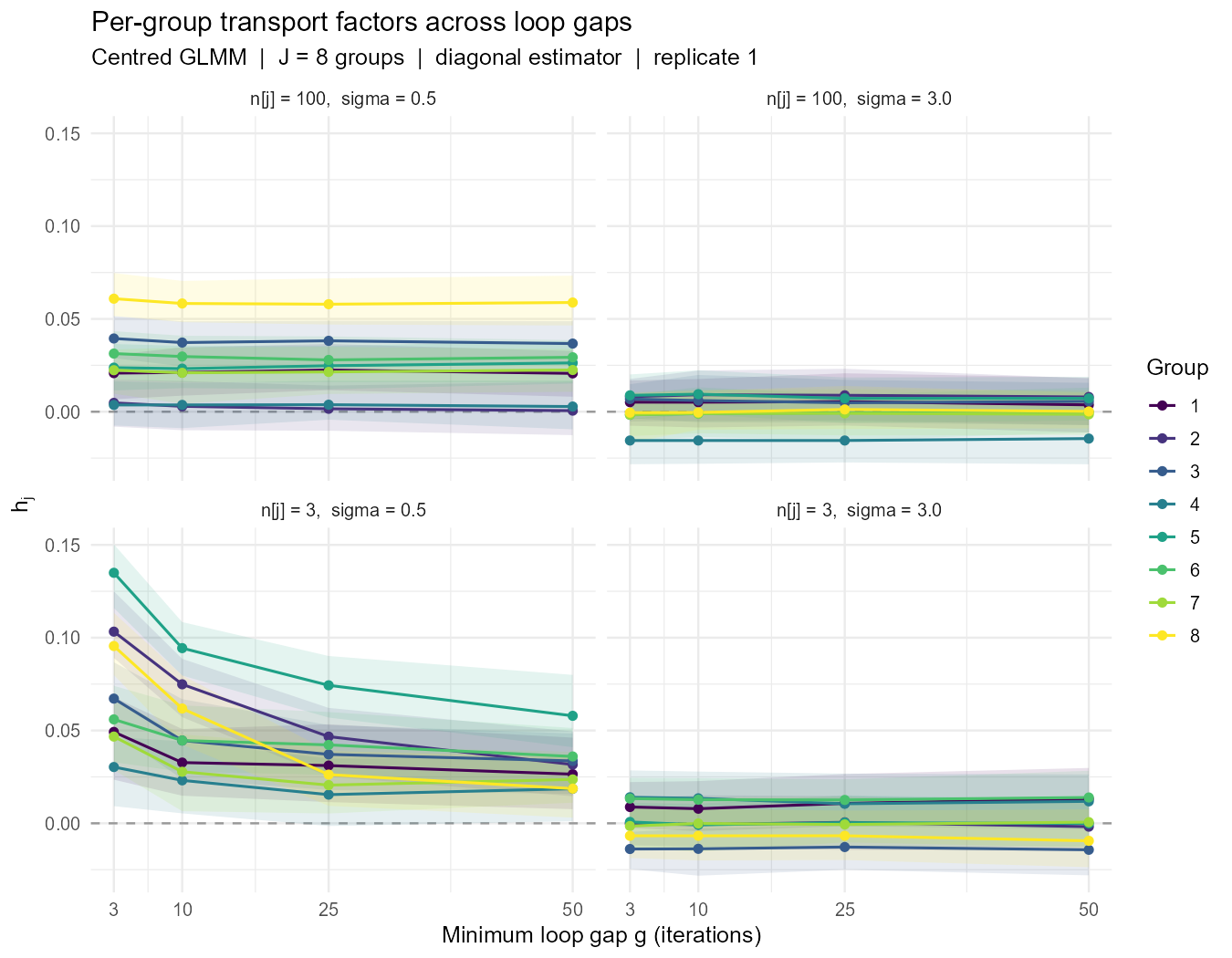}
  \caption{Per-group dependence coefficients $\rho_j$ as a function of the
    minimum loop gap $g \in \{3, 10, 25, 50\}$ for four corner cells
    (replicate 1).
    Shaded bands show 90\% bootstrap intervals.
    In the prior-dominated cell ($n_j = 3$, $\sigma_\text{true} = 0.5$,
    bottom left), $\rho_j$ decays monotonically from gap 3 to gap 50,
    consistent with an autocorrelation-dominated signal.
    In data-dominated cells ($n_j = 100$, top row), the coefficients are near
    zero for $\sigma_\text{true} \geq 1.0$; one group in the
    $n_j = 100$, $\sigma_\text{true} = 0.5$ panel shows a gap-independent
    $\rho_j \approx 0.06$, a candidate for residualisation leakage
    (OLS removes only the linear base-to-fiber dependence; nonlinear
    coupling at loop endpoints sharing the same base point can persist).
    The negative-control experiment of Section~\ref{sec:attribution}
    separates leakage from base-closure signal.
    The bottom-right cell ($n_j = 3$, $\sigma_\text{true} = 3.0$) also shows
    near-zero factors: even with sparse data, a wide prior reduces
    $1/\sigma^2$ and hence $\pi_j$, suppressing the signal.
    This illustrates that the coupling signal is driven by prior
    \emph{precision}, not by prior \emph{variance}.}
  \label{fig:gap_profiles}
\end{figure}

\begin{table}[ht]
  \centering
  \caption{Simulation study cell-level summary (median over 10 replicates,
    at headline gap $g = 50$).
    $\bar{\rho}$: mean dependence coefficient; $\rho_{\max}$: largest per-group
    coefficient; $\bar\pi$: mean prior fraction (analytic); \% Div: percentage
    of runs with at least one divergence.}
  \label{tab:summary}
  \small
  \begin{tabular}{rrcccr}
    \toprule
    $n_j$ & $\sigma_\text{true}$ & $\bar{\rho}$ & $\rho_{\max}$ & $\bar\pi$ & \% Div \\
    \midrule
    3 & 0.5 & 0.040 & 0.090 & 0.784 & 100 \\
    3 & 1.0 & 0.036 & 0.070 & 0.802 & 100 \\
    3 & 2.0 & 0.013 & 0.028 & 0.537 & 100 \\
    3 & 3.0 & 0.015 & 0.033 & 0.533 & 100 \\
    \midrule
    10 & 0.5 & 0.034 & 0.057 & 0.635 & 100 \\
    10 & 1.0 & 0.023 & 0.044 & 0.444 & 100 \\
    10 & 2.0 & 0.011 & 0.026 & 0.276 &  90 \\
    10 & 3.0 & 0.005 & 0.016 & 0.250 &  40 \\
    \midrule
    30 & 0.5 & 0.032 & 0.062 & 0.474 & 100 \\
    30 & 1.0 & 0.010 & 0.022 & 0.166 &  50 \\
    30 & 2.0 & 0.002 & 0.016 & 0.113 &   0 \\
    30 & 3.0 & 0.002 & 0.015 & 0.094 &   0 \\
    \midrule
    100 & 0.5 & 0.011 & 0.034 & 0.172 &  50 \\
    100 & 1.0 & 0.001 & 0.010 & 0.051 &  10 \\
    100 & 2.0 & 0.001 & 0.013 & 0.029 &   0 \\
    100 & 3.0 & 0.001 & 0.013 & 0.037 &   0 \\
    \bottomrule
  \end{tabular}
\end{table}

\section{What to Do About It: Guidance for Practitioners}
\label{sec:guidance}

The diagnostic identifies \emph{which} groups suffer a geometric obstruction
and \emph{why}; this section addresses what a practitioner should do next.
The good news is that the indicated remedies are well-established methods
that require no new sampler: the geometric analysis tells us \emph{when}
and \emph{where} to apply them.

\subsection{Per-group (partial) non-centering, guided by the prior fraction}
\label{sec:guidance_pnc}

The prior fraction $\pi_j$ \eqref{eq:prior_fraction} is precisely the
quantity that determines the optimal parameterisation of group $j$
\citep{papaspiliopoulos2003, papaspiliopoulos2007}.
Groups with $\pi_j$ near 1 (prior-dominated) should be non-centred; groups
with $\pi_j$ near 0 (data-dominated) should remain centred; intermediate
groups benefit from \emph{partial} non-centering,
\[
  \alpha_j = \mu + \sigma^{w_j}\,\tilde{\alpha}_j^{(w_j)},
  \qquad w_j \in [0, 1],
\]
which interpolates between the centred ($w_j = 0$) and non-centred
($w_j = 1$) coordinates.
For the Gaussian case the optimal partial non-centring weight is classical and
closed-form \citep{papaspiliopoulos2003}; for non-Gaussian models
\citet{gorinova2020} instead learn $w_j$ by stochastic variational optimisation
(the VIP algorithm). The geometric analysis of
Section~\ref{sec:glmm_connection} supplies the GLMM analogue in closed form:
$w_j \approx \pi_j$ evaluated at the posterior mean, recovering the variational
weight from a pilot run rather than an optimisation.
This is the primary practical payoff of the framework: a cheap, analytic,
per-group parameterisation rule with a clear justification.
We tested this rule directly on a mixed design.
Table~\ref{tab:mixed_design} reports median min-ESS (over all parameters,
$J = 8$ groups, 10 replicates) for the centred, half ($w_j = 0.5$),
non-centred, and $w_j = \pi_j$ parameterisations, across uniform-sparse, mixed,
and uniform-dense designs at $\sigma \in \{1, 2\}$.
Two findings emerge, both consistent with the two-pathologies view
(Section~\ref{sec:discussion}).

\emph{(i)}~On min-ESS the non-centred parameterisation wins in every cell.
The minimum is taken over all parameters and is dominated by $\sigma$, whose
centred geometry retains a funnel that non-centering removes regardless of group
size; in the range tested ($\sigma \in \{1,2\}$, $n_j \in \{3,50\}$)
non-centering is never harmful enough elsewhere to forfeit this advantage.
The $\pi_j$-adaptive rule interpolates in the right direction: at $\sigma = 1$
sparse, $\pi_j \approx 0.8$, so its weights are near-non-centred and it reaches
min-ESS $2381$, far above centred ($187$) and half ($1071$), but it does not
match full non-centring ($3327$), because the residual centering left by
$w_j < 1$ leaves some $\sigma$-funnel.
We therefore do not claim $w_j = \pi_j$ \emph{beats} both uniform choices: in
this regime it is a safe interpolator that improves on the worse uniform choice
and moves toward the better as $\pi_j$ rises, not a uniform winner.

\emph{(ii)}~The prior fraction nonetheless does its job at the level it speaks
to, the per-group fiber.
On per-group $\alpha_j$ ESS (Figure~\ref{fig:mixed_pergroup}), $w_j = \pi_j$
sits at or near the best method for every group: it non-centres exactly the
prior-dominated groups whose centred $\alpha_j$ ESS would otherwise lag, and
leaves data-dominated groups alone.
The min-ESS bottleneck is the $\sigma$-funnel, a base-space pathology that no
per-group fiber reparameterisation can touch; the gap between per-group success
and min-ESS is the two-pathologies distinction made concrete.
A practical aside visible in the same data: in the mixed design the data-rich
groups stabilise $\sigma$ and lift $\alpha_j$ ESS for \emph{all} methods, so a
few well-identified groups can make centring tolerable for the sparse ones.
The value of $w_j = \pi_j$ over a single uniform choice is thus robustness
across regimes (a uniform choice is optimal only within its own regime, and
very large $n_j$ or $\sigma \to 0$ would reverse non-centring's advantage, cases
not covered here) together with its role as an analytic per-group diagnostic
from a single pilot run.
This is the typical situation in unbalanced multilevel and panel data, where
groups carry very different numbers of observations; there a per-group choice
can matter, whereas balanced designs are usually well served by a single
uniform parameterisation.

\begin{table}[ht]
\centering
\caption{Controlled comparison: median min-ESS (over all parameters) by
  parameterisation and design. $J = 8$ groups, 10 replicates,
  $4 \times 2000$ post-warmup draws. $w_j = \pi_j$ is the proposed
  prior-fraction adaptive rule. Non-centred wins on min-ESS throughout because
  the bottleneck is the $\sigma$-funnel; $w_j = \pi_j$ interpolates toward it.}
\label{tab:mixed_design}
\small
\begin{tabular}{lrrrrrr}
\toprule
Method & \multicolumn{3}{c}{$\sigma = 1.0$} & \multicolumn{3}{c}{$\sigma = 2.0$} \\
\cmidrule(lr){2-4}\cmidrule(lr){5-7}
 & Sparse & Mixed & Dense & Sparse & Mixed & Dense \\
\midrule
  Centred ($w_j = 0$)        & 187  & 1988 & 1559 & 609  & 1993 & 1395 \\
  Half ($w_j = 0.5$)         & 1071 & 2143 & 1768 & 1457 & 2173 & 1510 \\
  Non-centred ($w_j = 1$)    & 3327 & 2626 & 1858 & 2349 & 2467 & 1628 \\
  $w_j = \pi_j$ (adaptive)   & 2381 & 2119 & 1621 & 1432 & 2082 & 1389 \\
\bottomrule
\end{tabular}
\end{table}

\begin{figure}[ht]
  \centering
  \includegraphics[width=0.95\textwidth]{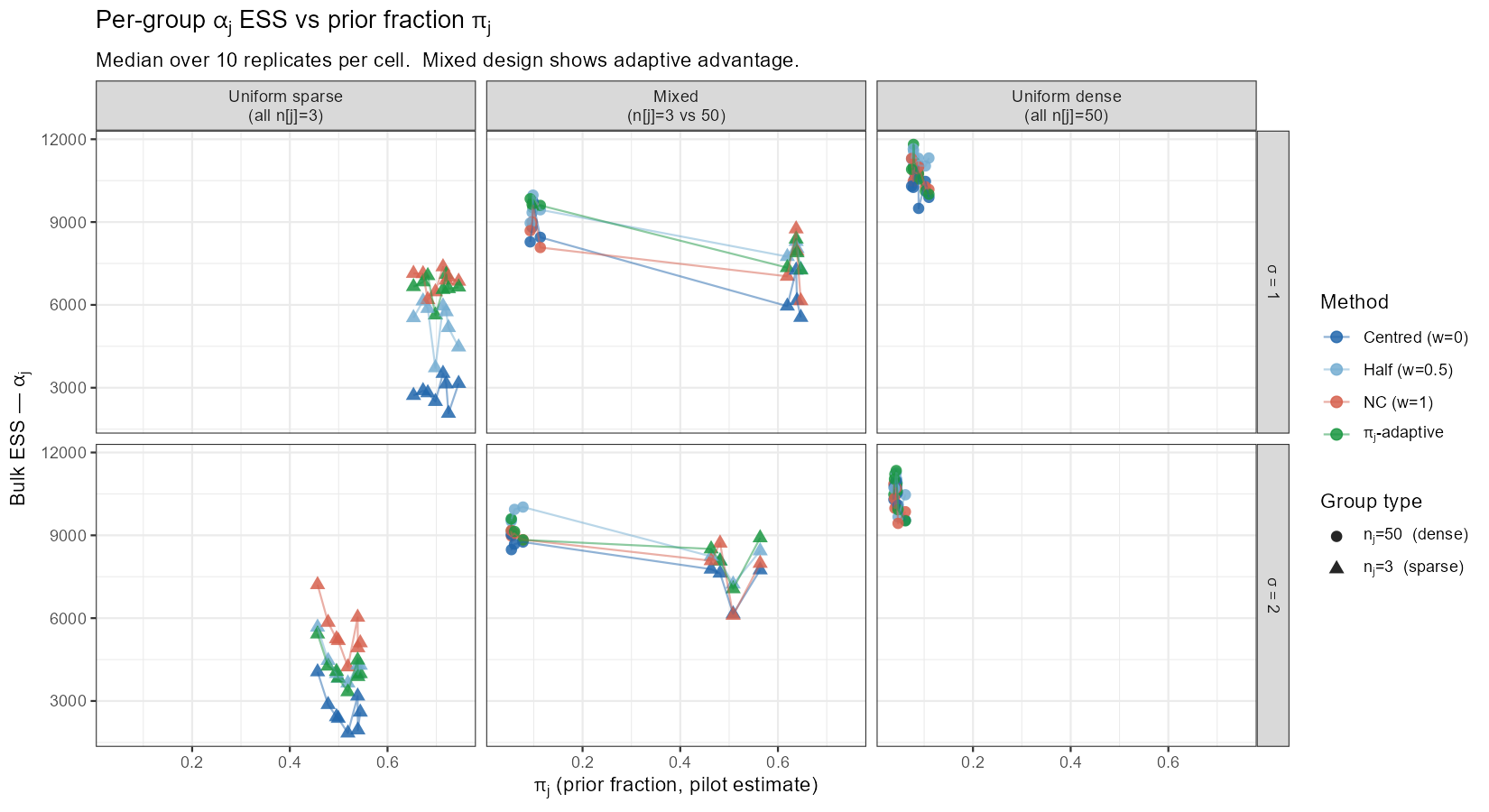}
  \caption{Per-group $\alpha_j$ ESS against the prior fraction $\pi_j$, by
    design (columns) and $\sigma$ (rows), median over 10 replicates.
    The $\pi_j$-adaptive rule (green) tracks the best method across the whole
    $\pi_j$ range: it lifts the $\alpha_j$ ESS of prior-dominated groups (high
    $\pi_j$), where centred (dark blue) lags, without harming data-dominated
    groups (low $\pi_j$).
    This per-group success coexists with non-centred winning on \emph{min}-ESS
    (Table~\ref{tab:mixed_design}), because the min-ESS bottleneck is the
    $\sigma$-funnel, not the fiber.}
  \label{fig:mixed_pergroup}
\end{figure}

\subsection{Interweaving (ASIS)}
\label{sec:guidance_asis}

When no single parameterisation works for all groups, or when the
practitioner prefers not to commit, the ancillarity--sufficiency
interweaving strategy of \citet{yu2011} alternates centred and non-centred
updates within each iteration.
ASIS is provably robust exactly in the regime the diagnostic
flags: where no single parameterisation suits every group simultaneously,
because the prior fractions $\pi_j$ are spread across the unit interval.
By combining the sufficient (centred) and ancillary (non-centred) augmentations
it is efficient whenever at least one is, group by group.
For Stan users, the per-group rule of Section~\ref{sec:guidance_pnc} is
usually simpler to implement; ASIS is the method of choice for Gibbs-style
samplers where both conditionals are tractable.

\subsection{Marginalising the fiber}
\label{sec:guidance_marginal}

The coupling pathology lives in the base--fiber dependence; it can be removed
entirely by integrating the fiber out.
For the logistic GLMM the conditional posterior of each $\alpha_j$ given
$(\mu, \sigma, \bm{\beta})$ is log-concave, so a Laplace or adaptive
quadrature approximation to the marginal likelihood of
$(\mu, \sigma, \bm{\beta})$ is accurate and cheap
\citep[as in INLA;][]{rue2009}.
Sampling the low-dimensional marginal and drawing
$\bm{\alpha} \mid \mu, \sigma, \bm{\beta}, y$ afterwards eliminates the
connection rather than correcting for it, at the cost of an approximation
whose error must be checked.
Where it applies, this is the most direct remedy; the diagnostic identifies
when the extra machinery is warranted.

\subsection{The connection as a fixed linear reparameterisation}
\label{sec:reparam}

The connection can also be turned into a reparameterisation directly.
Evaluating the coefficients $\{A_{j,\mu}, A_{j,\sigma}\}$ at the posterior mean
$(\mu_0, \sigma_0)$ and introducing horizontally-corrected fiber coordinates
\begin{equation}
  \tilde{\alpha}_j = \alpha_j
    - A_{j,\mu}\,(\mu - \mu_0)
    - A_{j,\sigma}\,(\sigma - \sigma_0)
  \label{eq:horiz_reparameterisation}
\end{equation}
removes the \emph{linear} contribution of the base to the fiber, weakening the
base-conditional dependence that $\pi_j$ measures.
This is \emph{not} the exact flat coordinate: by
Proposition~\ref{prop:universal_flat} that is the conditional score
$\partial_{\bm{\alpha}}\ell$, the precision-weighted residual
$(\alpha_j-\mu)/\sigma^2$ for the GLMM, which is nonlinear in the base.
Equation~\eqref{eq:horiz_reparameterisation} is its \emph{linearisation}, with
$A$ frozen at the posterior mean and kept to first order in
$(\mu-\mu_0,\sigma-\sigma_0)$.
It is implemented in Stan as \texttt{glmm\_hconnected.stan} and needs only the
connection coefficients as data.
Table~\ref{tab:ess_comparison} compares effective sample size for the centred,
non-centred, and horizontally-corrected models on the sparse GLMM benchmark
($J=8$, $n_j=3$, $\sigma_\text{true}=3$).

\begin{table}[ht]
  \centering
  \caption{ESS bulk (percentage of nominal 8{,}000) for key parameters under
    three parameterisations on the sparse GLMM.
    H-corrected uses the posterior-mean connection coefficients as in
    \eqref{eq:horiz_reparameterisation}.}
  \label{tab:ess_comparison}
  \small
  \begin{tabular}{lrrr}
    \toprule
    Parameter & Centred & H-corrected & Non-centred \\
    \midrule
    $\mu$ & 39.6\% & 90.6\% & 41.7\% \\
    $\sigma$ & 7.7\% & 12.3\% & 29.3\% \\
    $\alpha[1]$ & 66.5\% & 100.8\% & 100.2\% \\
    $\beta[1]$ & 60.8\% & 70.6\% & 86.4\% \\
    \bottomrule
  \end{tabular}
\end{table}

The H-corrected parameterisation dominates on $\mu$ (90.6\% vs.\ 39.6\% centred
and 41.7\% non-centred) and matches non-centred on $\alpha[1]$: the linear
correction removes the $\mu$--fiber coupling well, even though it is only the
linearisation of the exact trivialisation.
It is weaker on $\sigma$ (12.3\%) than non-centering (29.3\%), but for a reason
distinct from linearisation error: $\sigma$ is the funnel, a \emph{base-space}
pathology, and \emph{no} fiber reparameterisation (the exact score map
included) can touch it.
Since $\sigma$ is the bottleneck here, plain non-centering wins on minimum ESS;
the H-correction's gains are on $\mu$ and the fiber.
We therefore present it as a diagnostic-driven illustration that the connection
is directly actionable, not as a replacement for non-centering, the per-group
rule of Section~\ref{sec:guidance_pnc} is the tool we recommend.
The per-group rule is already used in the \texttt{smoothbp} change-point package
\citep{bindoff2026smoothbp}, where $\pi_j$ guides the parameterisation of each
subject's change-points; a worked example is in the package's reproduction
materials, with a fuller treatment in a companion paper.

\subsection{Connection-corrected proposals (proof of concept)}
\label{sec:guidance_horizontal}

Finally, the connection $A = -\GFF^{-1}\GBF$ is itself directly actionable:
it prescribes how the fiber should co-move when the base moves to remain on
the horizontal slice.
The \texttt{fibr} package provides \texttt{horizontal\_mcmc()}, a block
Metropolis-within-Gibbs sampler in which each base-block proposal
$(\mu, \sigma) \to (\mu', \sigma')$ pre-displaces the fiber by
$\Delta\bm{\alpha} = A(\mu, \sigma, \bm{\alpha})\,
(\Delta\mu,\,\Delta\sigma)^\top$
before the joint accept/reject step.
We include this as a proof of concept that the analytic connection is
actionable at negligible cost, not as a recommended production sampler;
for applied work the methods of
Sections~\ref{sec:guidance_pnc}--\ref{sec:guidance_marginal} are mature,
available in standard tooling, and should be preferred.
Two cautions are warranted.
First, naively applying the connection as a deterministic pre-displacement
without the corresponding Jacobian correction is biased: a simulation-based
calibration check of that variant fails decisively, whereas the
exactly-invertible Laplace transport and a reparameterised HMC in the
Laplace-standardised fiber both pass calibration.
This is consistent with flatness: the correct ``connection-aware'' move is a
reparameterisation, not a transport.
Second, on the sparse benchmark the calibrated variants are correct but do not
outperform non-centred NUTS in effective sample size per gradient evaluation;
they are a demonstration, not a faster sampler.
A full treatment of connection-aware sampling (horizontal leapfrog HMC
with $A$ recomputed at every leapfrog step, ergodicity analysis, and
multi-model benchmarking) is developed in a companion paper.

\subsection{Reading the prior fraction: influence and sensitivity}
\label{sec:guidance_prior}

Beyond the centring decision, $\pi_j$ is an inferential quantity: it is the share
of group $j$'s posterior precision supplied by the prior, so it localises
\emph{where prior choices can move conclusions}.
High-$\pi_j$ groups are prior-dominated and are the ones to interrogate;
low-$\pi_j$ groups are prior-robust by construction.
Two questions follow.
Is $\sigma$ well identified, or is $\pi_j$ inheriting the hyperprior?
Since $\pi_j = 1/(1 + I_j\sigma^2)$, with $I_j = G_{FF,j} - 1/\sigma^2$ the group's
own likelihood information, and $\sigma$ is identified by the \emph{number}
of groups rather than the data per group, a refit under an alternative
$\sigma$-prior shows whether the high-$\pi_j$ groups move (a hyperprior artefact)
or hold (a data statement).
And for those groups, is the population-level prior they are shrunk toward one you
would defend for a group with little data of its own?
A fuller treatment for applied users is left to a separate paper.

\begin{remark}[A frequentist reading]
\label{rem:frequentist}
The shrinkage is not itself Bayesian: the REML/BLUP predictor of $\alpha_j$ shrinks
by the same factor, and the flat-prior limit $\sigma \to \infty$ sends
$\pi_j \to 0$, recovering the per-group maximum-likelihood (no-pooling) estimate.
That limit is also where the sparse group becomes awkward --- with small $I_j$ the
no-pooling estimate has variance $\sim 1/I_j$ and is barely admissible, and such
groups are routinely discarded.
A proper hierarchical prior retains them with a coherent posterior, and $\pi_j$
states how much of the estimate rests on that prior rather than on the likelihood,
making the evidence explicit instead of excluding the case.
\end{remark}

\section{The \texttt{fibr} R Package}
\label{sec:package}

The \texttt{fibr} package \citep{bindoff2026fibr} centres on the prior fraction,
the quantity this paper concludes is the operative one; it is the maintained,
user-facing tool and is distributed on CRAN.
The geometric and sampler methods used to produce the figures (the analytic
connection, the loop-conditional dependence diagnostic, the transport integrator,
and the connection-corrected sampler) accompany the paper as reproduction code
under \texttt{paper/} in the source repository: with the connection flat for this
model class, they are demonstration apparatus rather than routine diagnostics, and
they load by sourcing \texttt{paper/setup.R}.
Table~\ref{tab:package} lists the package's exported functions.

\begin{table}[ht]
  \centering
  \caption{Exported functions in the \texttt{fibr} R package. The connection,
    dependence, and sampler methods of
    Sections~\ref{sec:diagnostic}--\ref{sec:guidance}
    (\texttt{compute\_connection()}, \texttt{holonomy\_diagnostic()},
    \texttt{integrate\_transport()}, \texttt{synthetic\_holonomy\_loop()},
    \texttt{horizontal\_mcmc()}) are reproduction code in the repository, not
    package exports.}
  \label{tab:package}
  \small
  \begin{tabular}{lp{0.65\textwidth}}
    \toprule
    Function & Description \\
    \midrule
    \texttt{prior\_fraction()} &
      Read-only per-coordinate prior fraction $\pi_j$ (shrinkage / pooling
      factor) for a fitted hierarchical model, with an adapter for
      \texttt{brms} fits and a manual path for other Stan models. Reports which
      group-level coordinates are prior-dominated, without reparameterising or
      refitting. \\
    \texttt{smoothbp\_advisor()} &
      Fisher information decomposition for changepoint random effects in
      \texttt{smoothbp} fits; returns per-group prior fractions and
      parameterisation recommendations. \\
    \bottomrule
  \end{tabular}
\end{table}

The package itself depends only on \texttt{posterior} \citep{posterior2022} for
chain handling and \texttt{ggplot2} for visualisation; the reproduction code under
\texttt{paper/} additionally uses \texttt{FNN} for $k$-nearest-neighbour search,
\texttt{Matrix}, \texttt{deSolve}, and a Stan backend
(\texttt{cmdstanr}, \citealp{cmdstanr}).
While the connection and dependence machinery target the two-level logistic
GLMM, the prior fraction is model-agnostic: \texttt{prior\_fraction()} computes
it per coordinate for any fitted hierarchical model, with an adapter for
\texttt{brms} fits, as a read-only report of which group-level estimates are
prior-dominated (mostly shrinkage) rather than data-driven.
Because the per-coordinate computation needs only the fitted variance components
and the per-observation working weights, it extends without special-casing to
nested, crossed, and multi-level designs and to non-Gaussian families. This is the
regime in which the intraclass correlation requires case-specific conventions (a
latent-scale residual for generalised models, a covariate value for random slopes,
separate adjusted and unadjusted forms at each level), and the package's tests
verify the per-coordinate values against an independent recomputation for nested,
crossed, and three-level Gaussian and Bernoulli fits.
Because it neither reparameterises nor refits, it carries none of the
calibration burden of a reparameterisation and serves as a lightweight
prior-influence diagnostic alongside $\hat{R}$ and effective sample size.
The package is available at \url{https://github.com/ABindoff/fibr}.

\section{Discussion}
\label{sec:discussion}

\subsection{Coupling vs.\ the funnel: two distinct pathologies}

The simulation study (Section~\ref{sec:simulation}) reveals two related but
distinct pathologies, neither of which is geometric holonomy of $\Aconn$.
The funnel is a base-space pathology, and it \emph{is} genuinely geometric: it
reflects the intrinsic curvature of the base, whose Fisher--Rao metric is
hyperbolic (connection (i) of Section~\ref{sec:which_connection}).
When $\sigma$ is large, the posterior of $\sigma$ is concentrated at small
values (shrinkage) and the conditional distribution $p(\bm{\alpha}|\mu,\sigma)$
is narrow, creating a highly non-isotropic target.
HMC struggles because leapfrog steps that work in the wide region overshoot
in the neck.
The flatness of $\Aconn$ says nothing about this curvature and does not remove
it.

The second is a base--fiber \emph{coupling} pathology: when $\sigma$ is
\emph{small} (tight prior, high prior precision), the connection
$A_{j,\mu} = 1/(\sigma^2 G_{FF,j})$ is large, so the conditional posterior of
$\alpha_j$ tracks $\mu$ tightly and the centred chain explores it slowly.
It is tempting to read this as holonomy, a sampler that does not know the
connection cannot make horizontally-corrected proposals, but
Proposition~\ref{prop:flat} shows the connection is flat, so there is no
holonomy to follow, and the key empirical finding
(Section~\ref{sec:attribution}) confirms it: NUTS does not \emph{transport}
the fiber, and chain displacements
$\alpha_\text{end} - \alpha_\text{start}$ do not track any
parallel-transport prediction.
What remains is purely statistical: \emph{excess conditional autocorrelation}.
Prior-dominated groups ($\pi_j$ large) retain loop-conditional fiber
dependence that likelihood-dominated groups do not, and a negative-control
experiment confirms this signal is tied to base-closure conditioning rather
than to lag autocorrelation alone.
The two pathologies can co-occur (sparse data and a wide prior create both a
funnel and, at intermediate $\sigma$, strong coupling), but they are
separable: the funnel is a base-space curvature problem, while the coupling is
captured entirely by the per-group prior fraction $\pi_j$.

\subsection{The connection is flat; curvature lives in the working metric}
\label{sec:working_metric}

Proposition~\ref{prop:universal_flat} settles a question an earlier draft of
this paper left open.
The metric-orthogonal connection is flat for \emph{any} smooth hierarchical
posterior, of any fiber dimension, with any prior --- not because of conjugacy
but because the fiber score is a first integral.
In particular, correlated random effects, spatial or temporal fiber dependence,
and vector-valued group effects with a non-diagonal prior covariance all give
\emph{flat} connections; they do \emph{not}, contrary to what one might expect,
produce rotational holonomy of the true connection.
The only ways the metric-orthogonal connection fails to be flat are (i) where
$G_{FF}$ is singular (the fiber-degeneracy regime discussed below, where the
foliation by score level sets breaks down and obstruction can genuinely
re-enter) and (ii) when the connection is built from a metric other than the
true Hessian.
Case (ii) is not a defect but an opportunity: it is where genuine curvature
lives.

Write $\Aconn_M = -M_{FF}^{-1}\GBF$ for the distribution orthogonal with respect
to a \emph{working} metric $M$: a fixed mass matrix, an empirical covariance,
or the prior-only metric.
The score argument no longer applies, since $M_{FF} \neq -\partial^2_{\bm\alpha}\ell$,
and $\Aconn_M$ is generically curved.
We verify this in a two-fiber example: a bivariate group effect $\alpha\in\R^2$
with base-dependent prior correlation
$\Lambda(\bm\theta) = R(\theta_2)\diag(e^{\theta_1},e^{-\theta_1})R(\theta_2)^\top$
and a logistic likelihood (so the posterior is non-Gaussian and $G_{FF}$ is
fiber-dependent).
The metric blocks are computed in closed form and validated against automatic
differentiation to machine precision ($\sim 10^{-16}$), and the holonomy
operator $H_M$ is obtained by integrating parallel transport with an adaptive
high-order solver (\texttt{DOP853}, tolerance $10^{-12}$); the true Fisher
connection returns to the identity at the integrator floor
($\norm{H_M-I}_F \sim 10^{-13}$) at every loop size, confirming
Proposition~\ref{prop:universal_flat} numerically.
Taking $M_{FF}$ to be the identity instead yields a holonomy operator with
\emph{complex} eigenvalues $0.986 \pm 0.167\,i$; an anisotropic fixed metric
$\diag(1,3)$ gives $0.998 \pm 0.071\,i$.
Crucially, $\norm{H_M - I}_F$ grows as the
\emph{first} power of the enclosed base area (fitted log-log slopes $1.08$ and
$0.99$, Figure~\ref{fig:working_metric}), the hallmark of genuine curvature
(leading-order holonomy $=$ curvature $\times$ area) rather than numerical
noise.
This is exactly the rotational holonomy that the true connection never produces
(Remark~\ref{rem:abelian}), and for which the full-matrix estimator
(\texttt{structure = "full"}) and the complex eigenspectrum become the right
tools; its magnitude measures how far the working metric is from horizontal.
(Reproduced with exact derivatives in \texttt{data-raw/ad\_holonomy.py}.)

\begin{figure}[ht]
  \centering
  \includegraphics[width=0.95\textwidth]{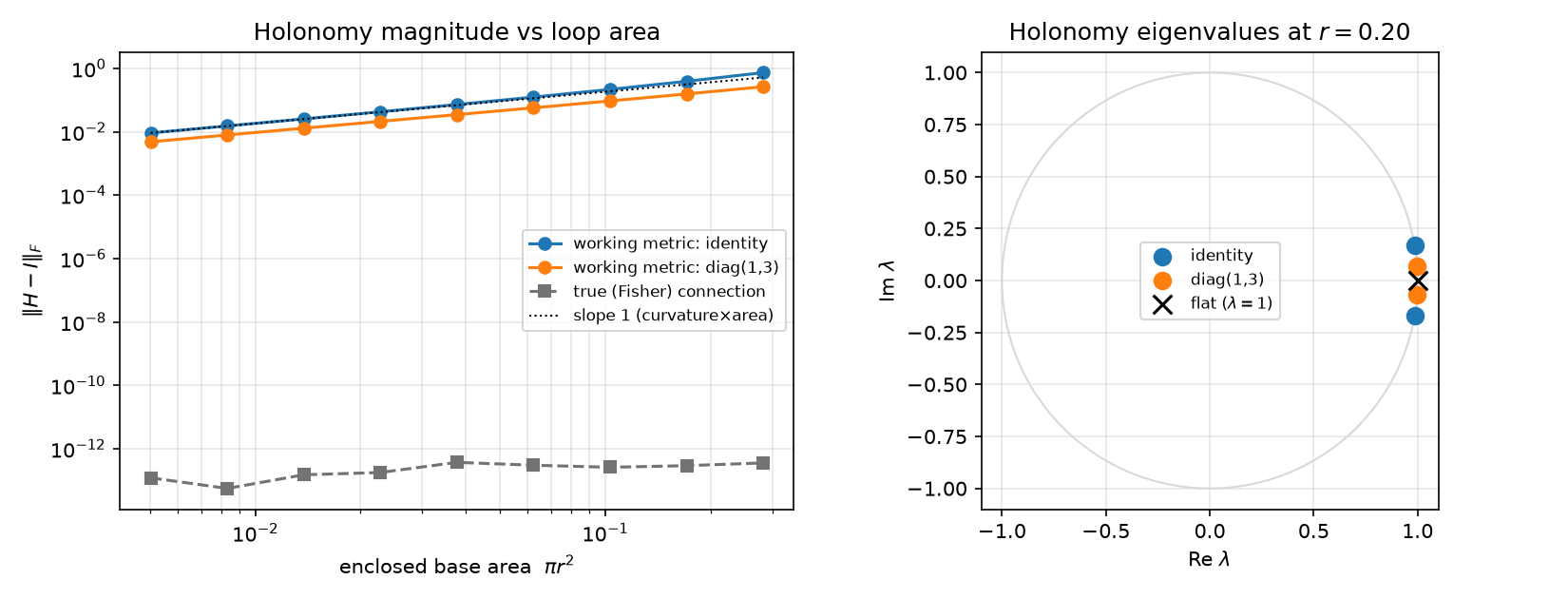}
  \caption{Working-metric connections are genuinely curved.
    \emph{Left}: Frobenius holonomy $\norm{H_M-I}_F$ against enclosed base area
    $\pi r^2$, log-log.
    The true (Fisher) connection sits at the integrator floor
    ($\sim 10^{-13}$, flat, Proposition~\ref{prop:universal_flat}); the
    identity and $\diag(1,3)$ working metrics lie on slope-1 lines, so their
    holonomy scales linearly with area --- the signature of a genuine
    curvature, not noise.
    \emph{Right}: holonomy eigenvalues at $r=0.20$; the working-metric spectra
    sit off the real axis (rotation), while the true connection gives
    $\lambda = 1$.
    Derivatives are validated against automatic differentiation to
    $\sim 10^{-16}$; transport uses \texttt{DOP853} at tolerance $10^{-12}$
    (\texttt{data-raw/ad\_holonomy.py}).}
  \label{fig:working_metric}
\end{figure}

This matters because every practical sampler uses a \emph{fixed} metric.
NUTS adapts a mass matrix once during warmup and then holds it; a blocked
sampler uses whatever conditional scale it is handed.
Such a sampler does not transport along the flat true connection but along the
curved $\Aconn_M$ of its own working metric, so it can experience a genuine
\emph{algorithmic} holonomy (a systematic, orientation- and area-dependent
fiber drift after the hyperparameters traverse a loop) even though the
underlying geometry is flat.
This reframes the empirical signal of Section~\ref{sec:diagnostic}: the
loop-conditional dependence $\hat\rho_j$ is an algorithmic, not geometric,
footprint, and the attribution test (Section~\ref{sec:attribution}) asks
precisely whether it carries the orientation and area structure that holonomy
would imply.
For NUTS on the GLMM it does not, which is consistent with NUTS adapting a
metric close enough to horizontal that $\Aconn_M$ has little curvature in the
explored region.
A deliberately mis-scaled sampler should instead show algorithmic holonomy
tracking the curvature of its working metric --- a concrete, testable
prediction.
This is the subject of a companion paper: a horizontal leapfrog HMC sampler that
recomputes $A$ at every leapfrog step so that trajectories follow horizontal
lifts, keeping the working metric aligned with $-\partial^2_{\bm\alpha}\ell$ so
that $\Aconn_M$ stays flat.
That paper develops the full treatment: ergodicity arguments, the
Metropolis--Hastings correction, a compiled backend, and multi-model benchmarks
against NUTS on coupling-dominated targets.

This closes the loop with the hypothesis we set out to test. For the geometric
program in MCMC, the Riemannian and connection-aware samplers that would naturally
pursue a curvature account of centring, the flat result is a redirection rather
than a dead end: there is no base--fiber holonomy to transport along in a
hierarchical posterior, so the effort is better spent keeping a sampler's working
metric aligned with the true Hessian, which is where the exploitable curvature
lives.

\subsection{Limitations and future work}

\paragraph{Model-specific derivations.}
The flatness itself is not model-specific (Proposition~\ref{prop:universal_flat});
only the \emph{closed-form} connection, curvature, and prior fraction of
Section~\ref{sec:glmm_connection} are particular to the logistic GLMM.
For other models the same quantities follow from $G_{FF}$ and $G_{BF}$, which
can be obtained by reverse-mode automatic differentiation of the log-posterior
gradient; supplying them is the natural next step for a compiled backend, and
the prior fraction $\pi_j = (\partial^2_{\alpha_j}\text{prior})/G_{FF,j}$
generalises directly.

\paragraph{Simulation study.}
The simulation study uses a single model family and does not include a
power analysis for the dependence diagnostic.
Characterising the minimum detectable signal as a function of chain length
and number of groups is an important practical question left for future work.

\paragraph{Fiber degeneracy.}
A second failure mode, distinct from residualisation leakage, arises when
the fiber itself has collapsed under the funnel pathology.
The diagnostic estimates per-group dependence coefficients from residual
fiber variation around detected loops; if the fiber has contracted to near a
point, as happens when an unidentified group-level parameter is pinned close
to its prior mean under a tight prior and near-absent likelihood, there is
no variance to autocorrelate and $\hat{\rho}_j \to 0$ regardless of pathology
severity.
The diagnostic is therefore blindest in the most severe cases: a collapsed
fiber and a well-mixed fiber are observationally equivalent in $\hat{\rho}_j$.
This was observed in a spike-and-slab changepoint model where some subjects'
second changepoint parameter was completely unidentified; those subjects
produced $\hat{\rho}_j \approx 0$ despite sitting in a Neal funnel, while the
analytic prior fraction $\pi_j$ correctly flagged the group.
The two failure modes --- residualisation leakage
(Section~\ref{sec:attribution}), which inflates $\hat{\rho}_j$ at short gaps,
and fiber degeneracy, which suppresses it in collapsed groups --- act in
opposite directions and should both be considered when interpreting a low
$\hat{\rho}_j$.
The analytic $\pi_j$, which depends only on the model structure and pilot
draws rather than on chain variation, is the more robust per-group flag and
should be treated as the primary output.

\paragraph{Prior dependence of $\pi_j$.}
We compute $\pi_j$ at a point estimate of $\sigma$, so it inherits the
$\sigma$-hyperprior when $\sigma$ is weakly identified (few groups). Reporting
$\pi_j$ per posterior draw, its distribution rather than a plug-in value,
propagates this uncertainty, and a wide posterior for $\pi_j$ is itself the signal
that a group's reading is hyperprior-sensitive. This refinement is natural future
work.

\paragraph{Variational inference.}
The same geometric structure applies to variational families: the fibration of
the space of variational parameters over the mean-field components has a
connection whose curvature determines the holonomy of the variational update.
This connection to variational Bayes geometry is speculative but worth pursuing.

\subsection{Note on methods of discovery}

This work was developed through an extended, iterative collaboration with Claude
Sonnet (Anthropic).
The fiber bundle framing, the mathematical derivations, the R package
architecture, and the exposition were developed jointly through dialogue rather
than by the author alone.
The mathematical results were checked both analytically and numerically (the
package test suite and \texttt{data-raw/verify\_flat\_connection.R}); the
flatness result in particular emerged from a critical re-examination of an
earlier draft that had reported only the linearised curvature.
The scientific framing, experimental decisions, and responsibility for errors
are the author's.
The author views this as an example of AI-collaborative mathematical research,
distinct from using AI as a writing assistant or code generator, and
reports it explicitly rather than subsuming it into generic tool-use language,
in the hope that transparent description contributes to evolving community
norms around human--AI co-production of scientific knowledge.

\section*{Acknowledgements}

Computations were performed using Stan \citep{stan2017} via cmdstanr
\citep{cmdstanr}, and the posterior package \citep{posterior2022}.

\bibliographystyle{plainnat}
\bibliography{fibr_paper}
\end{document}